\documentclass[a4paper, 12pt]{article}
\usepackage{amsmath, amsthm, amssymb, bm}
\usepackage[linesnumbered,ruled]{algorithm2e}

\usepackage{geometry}
\usepackage{graphicx}
\usepackage{tikz}
\usepackage{float}
\usepackage[margin=1cm]{caption}
\usepackage{cite} 
\usepackage{multirow}
\usetikzlibrary{automata,shapes.geometric}
\usepackage{array}
\usepackage{indentfirst}

\usepackage{authblk} 

 \usepackage{color}

\usepackage{hyperref}
\hypersetup{  colorlinks=true, linkcolor=blue, citecolor=red, urlcolor=blue  }
\usepackage[font=small,labelfont=bf]{caption}
\usepackage{physics}
\usepackage{algpseudocode}
\usepackage{algcompatible}
\allowdisplaybreaks 

\usepackage[titletoc,title]{appendix}

\usepackage{authblk}

\newtheorem{lemma}{Lemma}
\newtheorem{theorem}{Theorem}

\newtheorem{proposition}[theorem]{Proposition}

\renewcommand{\cal}[1]{\mathcal{#1}}

\newcommand{\Haf}{ \operatorname{Haf} }
\newcommand{\Per}{ \operatorname{Per} }
\newcommand{\z}{\vb{z}}
\newcommand{\x}{\vb{x}}
\newcommand{\X}{\vb{X}}
\newcommand{\s}{\vb{s}}
\renewcommand{\a}{\vb{a}}
\newcommand{\e}{\vb{e}}

\newcommand{\floor}[1]{\left\lfloor #1 \right\rfloor} 
\renewcommand{\binom}[2]{%
	\left(%
	\substack{	#1 \\	#2}
	\right)%
}   
\newcommand{\ba}{\begin{eqnarray}}
\newcommand{\ea}{\end{eqnarray}}

    \definecolor{airforceblue}{rgb}{0.36, 0.54, 0.66}
\definecolor{antiquebrass}{rgb}{0.8, 0.58, 0.46}

\newcommand*\samethanks[1][\value{footnote}]{\footnotemark[#1]} 

\begin{document}

\title{Speedup in Classical Simulation of Gaussian Boson Sampling}

\author[1,2]{Bujiao Wu\thanks{These two authors contributed equally} }

\author[3]{Bin Cheng\samethanks }
\author[1,2]{Jialin Zhang}
\author[3, 4]{Man-Hong Yung\thanks{yung@sustech.edu.cn}}
\author[1,2]{Xiaoming Sun\thanks{sunxiaoming@ict.ac.cn}}

\affil[1]{Institute of Computing Technology, Chinese Academy of Sciences, Beijing, China.}

\affil[2]{University of Chinese Academy of Sciences, Beijing, China.}

\affil[3]{Institute for Quantum Science and Engineering, and Department of Physics, Southern University of Science and Technology, Shenzhen 518055, China}

\affil[4]{Shenzhen Key Laboratory of Quantum Science and Engineering, Southern University of Science and Technology, Shenzhen 518055, China}

\date{}


%


\maketitle

\abstract

Gaussian boson sampling is a promising model for demonstrating quantum computational supremacy, which eases the experimental challenge of the standard boson-sampling proposal.
Here by analyzing the computational costs of classical simulation of Gaussian boson sampling, 
we establish a lower bound for achieving quantum computational supremacy for a class of Gaussian boson-sampling problems, where squeezed states are injected into every input mode. Specifically, we propose a method for simplifying the brute-force calculations for the transition probabilities in Gaussian boson sampling, leading to a significant reduction of the simulation costs. Particularly, our numerical results indicate that  we can simulate 18 photons Gaussian boson sampling at the output subspace on a normal laptop, 20 photons on a commercial workstation with 256 cores, and suggest about 30 photons for supercomputers. 
These numbers are significantly smaller than those in standard boson sampling, suggesting Gaussian boson sampling may be more feasible for demonstrating quantum computational supremacy.


\section{Introduction\label{introduce}}
There are a lot of efforts \cite{aaronson2017complexity,terhal2002adaptive,aaronson2011computational,bremner2010classical,boixo2018characterizing,kruse2018detailed} focusing on demonstrating the quantum supremacy~\cite{harrow2017quantum} over classical computers, such as sampling from commuting quantum circuits \cite{shepherd2009temporally}, random quantum circuits \cite{boixo2018characterizing}, and linear optical networks \cite{aaronson2011computational} as well as the 
variants \cite{lund2014boson, hamilton_gbs_2017}. 
All of the above sampling problems are shown to be intractable to classical computers, yet quantum devices nowaday are not large enough and their computational power is limited. At the same time of designing larger-scale quantum devices, classical simulation algorithms should also be explored to benchmark quantum supremacy and test the limit of classical computers.
For random circuit sampling, Boixo et al. \cite{boixo2018characterizing} claim that if the scale of quantum devices are greater than 50 qubits and depth 40, then it would be impossible for any current classical devices to simulate their specific random circuits. 
Since then, there have been many efforts trying to reach this limit, and eventually, the bar for quantum supremacy of random circuit sampling has been pushed higher~\cite{pednault2017breaking,boixo2017simulation,chen201864,li2018quantum,chen2018classical}. 

Boson Sampling (BS) is another well-known supremacy model, proposed by Aaronson and  Arkhipov \cite{aaronson2011computational}. However, the original boson sampling model is hard to implement experimentally since this model takes coherent single photons as input whereas it is difficult to generate a large number of photons at once.
Later, Lund et al. \cite{lund2014boson} improved it and proposed Scattershot Boson Sampling (SBS), which is argued to have the same complexity as BS. 
In the SBS model, single-photon source is replaced by squeezed lights, which is easier to prepare experimentally. Specifically, every mode of the linear optical network is injected with half of a two-mode squeezed state, whose another half is used for post-selection. 
SBS model was improved in Ref.~\cite{hamilton_gbs_2017}, where Gaussian Boson Sampling (GBS) model was introduced. GBS utilizes single-mode squeezed states (SMSS) as input, and there is no need for post-selection. The output probability of GBS is related to a matrix function called Hafnian, which, like permanent, is also $\# P$-hard to compute~\cite{bjorklund2018faster}.
This model is proved to be computationally hard based on the $\# P$ hardness of computing Hafnian~\cite{hamilton_gbs_2017, kruse2018detailed}. Moreover, there are a lot of applications for Gaussian boson sampling, such as finding dense subgraphs~\cite{arrazola2018using},  estimating the number of perfect matchings of undirected graphs~\cite{bradler2017gaussian}, molecular docking~\cite{banchi2019molecular}, etc.

In this paper, we introduce a classical simulation algorithm for GBS, with time complexity $O(m \sinh^2 r + \text{poly}(n) 2^{8n/3})$, where $m$ is the number of modes, $n$ is the number of output photons and $r$ is the squeezing parameter. If $O(m2^{n})$ memory is used to store and recycle the intermediate calculations, the time complexity can be improved to $O(m \sinh^2 r + \text{poly}(n) 2^{5n/2})$.
In our simulation model, each mode of the linear optical network is injected with a SMSS with identical squeezing parameter. The source code of our algorithm is provided in Ref.~\cite{Code}.

Our simulation algorithm can be divided into two steps. In the first step, we sample $n$ photons with time polynomial in the number of modes $m$.
In the second step, we sample an output configuration $\s =(s_{1},\cdots, s_{m})$ for some fixed $n$, where $s_{j}$ is the photon number in the $j$-th mode. The second step uses a similar idea of the algorithm in Ref.~\cite{clifford2018classical}, but the form of marginal probabilities in our case is more complicated, since calculating Hafnian might be harder than calculating permanent~\cite{bjorklund2018faster}. Specifically, we give a method to break a large Hafnian (and the square of its modulus) into pieces of smaller Hafnians and permanents, so that we can speed up the calculation of marginal probabilities. In this way, we can sample a configuration much more efficiently than the brute-force sampling.
For the second step, we can simulate 18 photons on a laptop, and 20 photons on the HuaWei Kunlun server with 256 cores, in about one day. As a comparison, using the brute-force method, we can only sample 6 photons on a laptop. Based on the simulation result, Sunway TaihuLight is estimated to be able to sample 30 photons.

There are some other classical simulations for the boson sampling problems. For standard boson sampling, Neville et al.~\cite{neville2017classical} performed a simulation for sampling 30 photons approximately on a laptop and 50 photons on a supercomputer, based on Metropolised independence sampling algorithm, restricted on the collision-free regime (i.e. the regime where the probability of observing more than one photons in one mode is sufficiently small). At the same time, Clifford et al.~\cite{clifford2018classical} introduced an $O(n2^{n})$ time classical simulation algorithm to produce one sample of the output distribution of BS. On the other hand, there are also works considering the simulability of boson sampling under noise~\cite{shchesnovich2019noise,oszmaniec2018classical}.

As for GBS, Quesada et al. introduced a GBS model with threshold detector~\cite{quesada2018gaussian}, which only detects whether the photons exist or not in one mode, and gave an $O(m^{2}2^{n})$ time and exponential space classical sampling algorithm for $m$ modes and $n$ clicks, where a click happens when there are photons in one mode. Recently, they also performed an actual implementation of their algorithm on the Titan supercomputer~\cite{gupt_classical_2018}. They can simulate a problem instance with 800 modes and 20 clicks using 240,000 CPU cores in about 2 hours. Thus, by Table \ref{tab:ComTime}, in the collision-free regime, where GBS with threshold detector is equivalent to GBS with photon-number-resolving detector, our algorithm has a better performance.



 \begin{table}[]
    \centering
    \begin{tabular}{c|c|c |c}
        \hline\hline
       \multicolumn{2}{c|}{\multirow{2}{*}{classical simulation Alg.}}  & \multicolumn{2}{c}{ classical limit ($n, m, t $)}\\
       \cline{3-4}
        \multicolumn{2}{c|}{} & non-parallel & super-computer\\
         \hline
        \multicolumn{2}{c|}{ BS\cite{neville2017classical}} & (30, 900, 0.5) & (50, 2500, $2400$) (prediction for Tianhe) \\
        \hline
        \multirow{3}{*}{GBS} &using threshold detectors\cite{quesada2018gaussian} & (12 ,288, 21.86) & (20, 800, 1.83) (Titan)\\
        \cline{2-4}
        & non-collision space & (18, 648, 5.7) & (30, 1800,$9.5$)(prediction for Sunway)\\
        \cline{2-4}
        & full space & (18, 324, 10.5) & (30, 900, 4.78)(prediction for Sunway)\\
        \hline\hline
    \end{tabular}
    \caption{Comparison the classical simulation limit of several different setups of linear optic network experiment, where $n, m, t$ denote the number of photons, the number of modes, and the corresponding time, with the unit of hours respectively.}
    \label{tab:ComTime}
\end{table}

This paper is organized as follows. In section~\ref{def}, we review the GBS model and Hafnian problem. In section~\ref{simulation}, we introduce our algorithm, and give numerical results of our algorithm. In section~\ref{conclusion}, we give a conclusion.

\section{Overview of Gaussian Boson Sampling\label{def}}
In this section, we review the definition of Gaussian boson sampling (GBS). Specifically, Section~\ref{subsec:defPhoton} reviews the process of generating $n$ photons with single-mode squeezed states (SMSS). Section \ref{subsec:defConf} reviews the output distribution for basic GBS model. Section \ref{subsec:defHaf} reviews the definition of Hafnian problem and introduce some properties of the Hafnian of some specific matrix which will be used in later section. 

\subsection{Probability of generating $n$ photons \label{subsec:defPhoton}}
GBS refers to the procedure of sampling photons from a linear optical network supplied with Gaussian input states~\cite{kruse2018detailed}, including SMSS as a special case. 
A SMSS $\ket{\phi}$ can always be expanded in the Fock basis with an even number $s$ of photons~\cite{gerry2005introductory}:
\ba
\ket{\phi} = \sum_{s \text{ even}} c_s(r) \ket{s} \ ,
\ea
where the coefficients $|c_s(r)|^{2}= \frac{1}{\cosh{(r)}} \left( \frac{\tanh{(r)}}{2} \right)^{s} \frac{s!}{((s/2)!)^{2}}$ depend on the squeezing parameter $r$.

Suppose SMSS with identical squeezing parameter $r$ is injected into all $m$ modes of the optical network, we use $m$ to denote the number of modes of optical network. The input state is then given by, 
\ba
(\ket{\phi})^{\otimes m} = \sum_{\vb{s}} c_{\vb{s}} (r) \ket{\vb{s}} \ ,
\ea
where $\vb{s} = (s_1, \ldots, s_m)$ represents one of the configurations with $s_j$ being the photon number in the $j$-th mode, and $c_{\vb{s}} := c_{s_1}\cdots c_{s_m}$. 

Let us denote $\ket{\bar{n}} := \sum_{s_1 + \cdots + s_m = n} c_{\vb{s}} \ket{\vb{s}}$ to represent the (unnormalized) superposition of all input configurations of $n$ photons. Then we can also write the input state as, 
\ba
(\ket{\phi})^{\otimes m} = \sum_{n} \ket{\bar{n}} \ ,
\ea
where the normalization of $\ket{\bar{n}}$ equals to the probability of generating $n$ photons from $(\ket{\phi})^{\otimes m}$, denoted as $P_n$, i.e.,
\ba
P_n := \ip{\bar{n}} = \sum_{s_1 + \cdots + s_m = n} |c_{\vb{s}}|^2 \ .
\ea

\subsection{Probability of photon configuration\label{subsec:defConf}}

Next, suppose the action of the optical network is denoted by a unitary transformation $U$, the output state is given by $U (\dyad{\phi})^{\otimes m} U^\dagger$. The output probability $p(\vb{s})=p(s_1, \cdots, s_m) =\Tr[U (\dyad{\phi})^{\otimes m} U^\dagger \dyad{\vb{s}}]$ for measuring a particular configuration $\ket{\vb{s}}$ can be expressed as,
\ba
p(\vb{s}) &=& \sum_{n, n'} \Tr[U \dyad{\bar{n}}{\bar{n}'} U^\dagger \dyad{\vb{s}}] \\
&=& \sum_{n, n'} \mel{\vb{s}}{U}{\bar{n}} \mel{\bar{n}'}{U^\dagger}{\vb{s}} \ .
\ea
Since optical transformation preserves photon number, the summand is non-zero only if $n = n' = s_1+\cdots +s_m$. Consequently, we can compactly write $p(\vb{s}) = |\mel{\vb{s}}{U}{\bar{n}}|^2$. Suppose we further define 
\begin{equation}
p_n(\vb{s}) := \frac{|\mel{\vb{s}}{U}{\bar{n}}|^2}{\ip{\bar{n}}}
\end{equation}
to be the probability of measuring the configuration $\ket{\vb{s}}$ if $n$ photons are generated from $(\ket{\phi})^{\otimes n}$, then we can obtain the following:
\begin{align}
p(\vb{s}) = \ip{\bar{n}} \frac{|\mel{\vb{s}}{U}{\bar{n}}|^2}{\ip{\bar{n}}} = P_n \cdot p_n(\vb{s}) \ .
\end{align}
Note that $p_n(\vb{s})$ is similar to that of the standard boson sampling (boson sampling with Fock state input), except that the input state $\ket{\bar{n}}$ is however a superposition of Fock states with $n$ photons.


\subsection{Probability in terms of Hafnian\label{subsec:defHaf}}
On the other hand, the output probability $p(\vb{s})=p(s_1, \cdots, s_m) $ for each configuration~$\ket{\vb{s}}$ can be expressed explicitly through Hafnian~\cite{kruse2018detailed},
\begin{align}
p(s_1, \cdots, s_m) &= \frac{\tanh^{n}(r)}{ s_{1}!\cdots s_{m}! \cosh^{m}(r)} |\text{Haf} (W_{\vb{s}})|^{2} \ ,
\label{ProOri}
\end{align}
Here the symbol `Haf' stands for Hafnian, which is a matrix function similar to permanent and determinant. The Hafnian of a symmetric $n\times n$ matrix $V$ is defined by ($n$ must be even) 
\begin{align}
\text{Haf}(V):=\sum_{\sigma\in \mathcal{M}_{n}}\prod_{j = 1}^{n/2} V(\sigma_{2j - 1},\sigma_{2j}) \ ,
\label{HafMatch}
\end{align}
where $\mathcal{M}_{n}$ is the set of all perfect matchings of $[n] := \{ 1, 2, \ldots, n \}$ and $V(i, j)$ is the $(i, j)$-th element of $V$. 
For example, when $n$ equals to 4, $\cal{M}_4 = \{ (12)(34), (13)(24), (14)(23) \}$, where $(ij)$ is a matching pair, and thus
\begin{align}
\Haf(V) = & V(1,2)V(3,4) + V(1,3)V(2,4) + V(1,4)V(2,3) \ .
\end{align}
Hafnian is permutation invariant from its definition, that is we can first interchange two columns of $V$, and then interchange the corresponding rows, and the Hafnian of the new matrix will remain the same.
Similar to permanents, exact computation of matrix Hafnian is also $\#P$-hard~\cite{valiant1979complexity}, which also implies GBS for the same squeezing parameters is also computational hard; the best classical algorithm takes time $O(n^3 2^{n/2})$~\cite{bjorklund2018faster}, which is closely related to the hardness of GBS~\cite{kruse2018detailed}.

Next, the matrix $W_{\vb{s}}$ is constrcuted from another $m \times m$ matrix (here $U$ refers to the $m\times m$ unitary acting on the single-photon subspace)
\begin{equation}
W := U U^t \ ,
\end{equation}

in the following way:
first, we take $s_j$ copies of the $j$-th column of $W$ to form a $m \times n$ matrix, i.e., with $m$ rows and $n$ columns. Then, we take $s_j$ copies of the $j$-th row of the $m \times n$ matrix to form $W_{\vb{s}}$, which is a $n\times n$ matrix. For example, suppose $W := (w_{ij})$ is a $6\times 6$ matrix and $\vb{s} = (1,1,0,2,0,0)$. Then 
\ba
\label{eq:W_s_example}
W_{\vb{s}} = 
\begin{pmatrix}
w_{11} & w_{12} & w_{14} & w_{14} \\
w_{21} & w_{22} & w_{24} & w_{24} \\
w_{41} & w_{42} & w_{44} & w_{44} \\
w_{41} & w_{42} & w_{44} & w_{44}
\end{pmatrix} \ .
\ea




\section{Classical simulation algorithm\label{simulation}}

In this section, we will describe our classical algorithm in details.
Our classical algorithm consists of two steps. The first step is sampling $n$ photons according to $P_n$, and the second step is sampling the configuration $\vb{s}$ according to $p_n(\vb{s})$. Although the time complexity of our sampling algorithm is exponential, it has great advantage compared to brute force sampling algorithm. We also give some numerical results for our algorithm as well as an estimation of the classical limit.

\subsection{Step 1: Sampling $n$-photon distribution}

\begin{figure*}
\center
\includegraphics[trim = 0mm 45mm 0mm 48mm, clip=true,width = 1.0\textwidth]{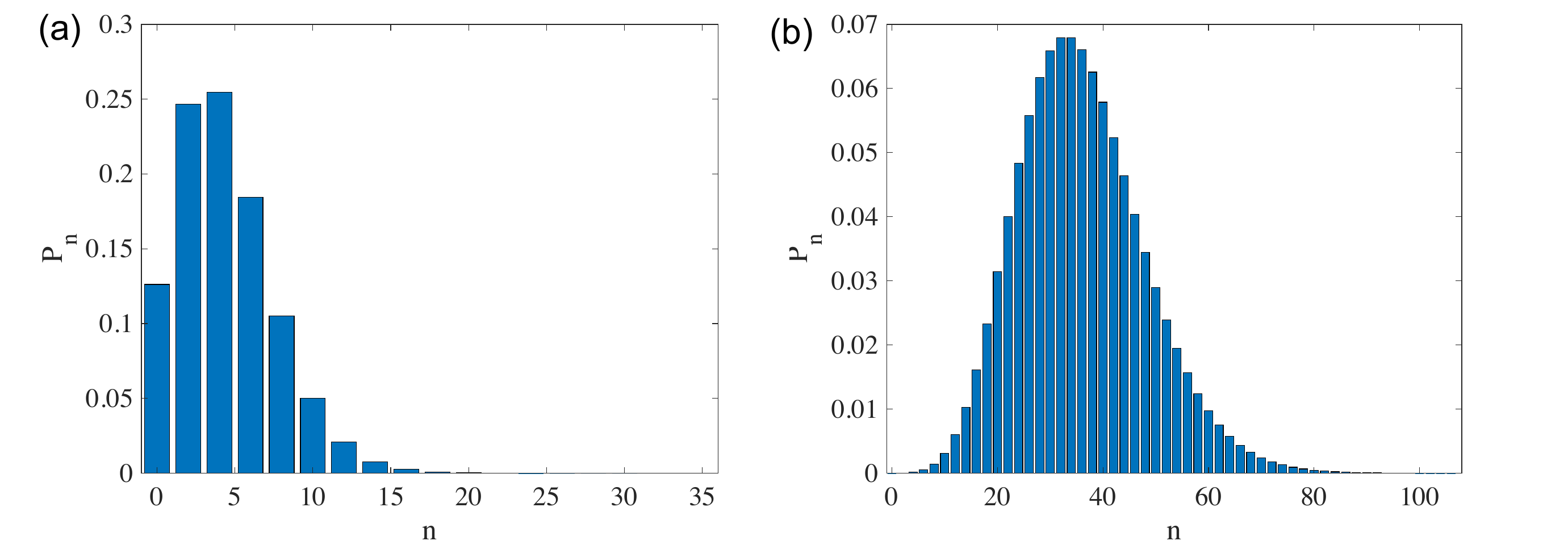}
\caption{Probabilities of generating $n$ photons ($n$ is even) from $m = 36$ single-mode squeezed states with a squeezing parameter (a) $r = 0.3423 \approx \ln (1 + 1/\sqrt{6}) $ and (b) $r = 0.8814 \approx \ln(1 + \sqrt{2}) $. }
\label{DisPhotons}
\end{figure*}

Recall that $P_n=0$ for odd $n$. For even $n$, the distribution $P_n$ is given by the negative binomial distribution~\cite{hilbe2011negative,kruse2018detailed},
%
\begin{align}
P_n = \binom{\frac{m+n}{2}-1}{\frac{n}{2}} \frac{\tanh^{n} (r)}{\cosh^{m} (r)} \ ,
\label{Eqdistribution}
\end{align}
and $\sum_{n=0}^{\infty} P_{n} = 1$. In principle, one needs to include $n$ from $0$ to infinity, i.e., $\{0,2,4,\cdots,\infty\}$. However, the probability of measuring photon numbers that is much larger than the most probable photon number~\cite{kruse2018detailed} $n_{\text{most}} = 2\floor{ (m/2 - 1)\sinh^{2} r }$ is very small.  In other words, one can consider a cutoff value given by $c \, n_{\text{most}}$  for some constant $c$. For example, the most probable number of photons is about $\sqrt{n}$ when we set $r = \ln (1 + \sqrt{1/n})$, and about $m$ when we set $r = \ln (1 + \sqrt{2})$~\cite{kruse2018detailed}. {See Figure \ref{DisPhotons} for an example with $m = 36$ modes, and the squeezing parameters are about $\ln (1 + 1/\sqrt{6})$ and $\ln (1 + \sqrt{2})$ for Figure \ref{DisPhotons} (a) and (b) respectively. }
In the following, we will impose a truncation to the distribution $\{ P_n \}$, sampling from $n \in \{0,2,4,\cdots, N\}$, where $N := c n_{\text{most}}$. For each $n$ in this set, we calculate $P_n$, which can be done in polynomial time. After that, we renormalize $\{ P_0, P_2, \cdots P_N \}$ to make the truncated set a real probability distribution. Since there are $N/2$ probabilities, we can sample an even $n$ in $O(n_{\text{most}})$ time, which is linear in $m$.

\subsection{Step 2: change of basis}

Now, our task is to sample an $\s$ from distribution $p_n(\s)$. From Eq.~\eqref{ProOri} and \eqref{Eqdistribution}, one can get the explicit form of $p_n(\s)$, which becomes independent of the squeezing parameter~$r$:
\ba
p_n(\s) = \frac{p(\s)}{P_n} = \frac{1}{s_1!\cdots s_m! \binom{\frac{m+n}{2} - 1}{\frac{n}{2}}} |\Haf(W_{\s})|^2 \ .
\ea

Previously, we are working with second quantization in the Fock space. But for the purpose of our algorithm, we will change to first quantization. In other words, we consider the photons states are symmetrized in a certain basis as follows. Denote $\ket{\x} := \ket{x_1 x_2 \cdots x_n}$ to be the position basis of photons, where $x_i \in \{ 1,...,m\}$ denotes the output mode of the $i$-th photon. For identical particles like photons, we will have the following relation,
\ba
\ket{\s} = \frac{1}{\sqrt{N_p}} \sum_{\x} {}^{'} \ket{\x} \ ,
\ea
where the summation $\sum {}^{'}$ is over those $\ket{\x}$'s that match the configuration $\ket{\s}$ and $N_p := n!/(s_1!\cdots s_m!)$ is the number of terms in the summation. For example, suppose $\s = (1,1,0,2,0,0)$ and then the corresponding $\x$ can be $(1,2,4,4)$, $(2,1,4,4)$, $(4,1,2,1)$, etc.

Then, the output state $U (\ket{\phi})^{\otimes m} = \sum_{\s} \tilde{c}_{\s} \ket{\s}$ can be written in the position basis,
\ba
\sum_{\s} \frac{\tilde{c}_{\s}}{\sqrt{N_p}} \sum_{\x} {}^{'} \ket{\x} = \sum_{\x} \frac{\tilde{c}_{\s}}{\sqrt{N_p}} \ket{\x} \ ,
\ea
where $\tilde{c}_{\s} := \mel{\vb{s}}{U}{\bar{n}}$. Note that $q(\x) := |\tilde{c}_{\s}|^2/N_p$ is the probability of measuring $\ket{\x}$. Similar to $p_n(\s)$, we can define a conditional probability $q_n(\x) := q(\x)/P_n$, whose explicit form is,
\ba
q_n(\x) = \frac{1}{f_n} |\Haf(W_{\s})|^2 \ , 
\ea
where $f_n := n! \binom{\frac{m+n}{2} - 1}{\frac{n}{2}}$. 
{Appendix \ref{Normalization} also give the proof that $q_n(\x)$ indeed form a probability distribution.}
One can see that if $\ket{\x}$ and $\ket{\x'}$ correspond to the same configuration $\ket{\s}$, $q_n(\x) = q_n(\x')$.

In the expression of $q_n(\x)$, the Hafnian part is a function of $\s$, but we can change it to a function of $\x$. Define $W_{\x}$ by 
\begin{equation}
W_{\x}(i, j) := W(x_i, x_j) \ .
\end{equation}
Note that $W_{\x}$ is symmetric, because $W$ is symmetric as can be easily verified from its definition. 
In Appendix~\ref{App:equivalence}, we will show that although $W_{\x} \neq W_{\s}$, their Hafnians are the same: $\Haf(W_{\x}) = \Haf(W_{\s})$. Therefore, we have
\ba
q_n(\x) = \frac{1}{f_n} |\Haf(W_{\x})|^2 \ . \label{EqDisE}
\ea
Moreover, this fact yields $p_n(\s) = q_n(\x) N_p$, which implies that we can sample an $\x$ from $q_n(\x)$, and then read an $\s$ from $\x$. Since there are $N_p$ number of $\x$'s that give the same $\s$, the probability of sampling an $\s$ with this method is $q_n(\x) N_p$, that is $p_n(\s)$.

\subsection{Step 3: sampling from marginal distribution}

We now focus on sampling from distribution $q_n(\x)$. Similar to Ref.~\cite{clifford2018classical}, one can always decompose the distribution $q_n(\x)$ as follows,
\begin{align}
q_n(x_{1},\cdots,x_{n}) = & q_n(x_{1}) q_n(x_{2}|x_{1}) \notag \\ 
& \cdots q_n(x_{n}|x_{1},\cdots,x_{n-1}) \ ,
\label{ConPro}
\end{align}
which implies that we can obtain a sample $(x_{1},\cdots,x_{n})$ by sampling sequentially each conditional distribution given by,
\begin{align}
q_n(x_{k}|x_{1},\cdots,x_{k-1}) = \frac{q_n(x_{1},\cdots,x_{k})}{q_n(x_{1},\cdots,x_{k-1})} \ ,
\end{align}
where the marginal probability is defined by 
\begin{equation}
{q_n}\left( {{x_1}, \cdots ,{x_k}} \right): = \sum\limits_{{x_{k + 1}}, \ldots ,{x_n}} {{q_n}\left( {{x_1}, \cdots ,{x_n}} \right)} \label{eq:marginal_definition} \ .
\end{equation}
In other words, to sample from the distribution $q_n(\x)$, our main task is to evaluate many marginal probabilities $q_n(x_{1},\cdots,x_{k})$. Explicitly, after obtaining a sequence, $(x_1,x_2,\cdots,x_{k-1} )$, to sample the position of the $k$-th photon, $x_k$, we need to evaluate $m$ conditional probabilities $q_n(x_{k}|x_{1},\cdots,x_{k-1})$ for all $x_k \in [m]$ \footnote{$[m]$ represent set $\{1, \ldots, m\}$. }. In total, there needs to compute $mn$ conditional probabilities (or marginal probabilities) to sample a full string $\x$.  

From Eq.~\eqref{EqDisE}, each marginal probability involves a total of $m^{n-k}$ terms,
\begin{align}
q_n(x_1, \cdots, x_{k}) \propto & \sum_{x_{k+1}, \ldots , x_n} |\Haf(W_{\x})|^2 \label{eq:marginal_haf} \\ 
=& \sum_{x_{k+1}, \ldots , x_n} \sum_{\sigma, \tau \in \cal{M}_{n}} \prod_{j = 1}^{n/2}  W(x_{\sigma_{2j-1}}, x_{\sigma_{2j}}) W^*(x_{\tau_{2j-1}}, x_{\tau_{2j}}) \ . \label{eq:marginal} 
\end{align}
Here, we use $\sigma$ to label the variables appearing in $W$ and $\tau$ to label those in $W^*$. 
According to Ref.~\cite{bjorklund2018faster}, the time complexity of computing the Hafnian of a $n\times n$ matrix is $O(n^{3}2^{n/2})$, so computing $q_n(x_1, \cdots, x_k)$ by directly computing $m^{n-k}$ $\Haf(W_{\x})$ is of course costly. However, we shall show that one can eliminate many terms by breaking each Hafnian $|\Haf(W_{\x})|^2$ into many pieces of smaller sub-Hafnians. In this way, one can achieve a significant speedup compared with the brute-force sampling algorithm (see Section~\ref{subsec:numerical} for a comparison). 

In summary, in our algorithm \ref{alg:GauSample1}, the time complexity for one sample is ${O}(\text{poly}(n) 2^{8n/3})$ (Appendix \ref{timeA}), if the intermediate values of the sub-Hafnians are not recycled. If we use $O(m2^{n})$ space to store those sub-Hafnians, we can avoid a lot of repeated computations, reducing the runtime to $O(\text{poly}(n) 2^{5n/2})$ (Appendix \ref{timeB}). For our purpose, we would call the former `polynomial-space algorithm' and the latter `exponential-space algorithm'. Alg. \ref{alg:GauSample1} shows a short version of the pseudocode of our classical simulation algorithm (see Appendix \ref{AlgLemma} more details).
\begin{center}
\begin{minipage}{1\linewidth}
\IncMargin{0.1cm}
\begin{algorithm}[H]
\SetKwInOut{Input}{input}
\SetKwInOut{Output}{output}
\Input{$m,n$ and an $m\times m$ matrix $U$.}
\Output{a sample $\mathbf{s}$.}
$W:= UU^{T},\mathbf{x}:= \emptyset$\;
\For {$k := 1$ to $n$}{
    \For{$l := 1$ to $m$}{
        \emph{Compute $w_{l}$ with Theorem~\ref{thm:marginal}} \tcp*{$w_l$ is the marginal}
    }
    $\mathbf{w}_k = (w_{1},\cdots, w_{m})$ \;
    $x_k \leftarrow Sample(\mathbf{w}_k)$ \tcp*{Sample an $x_k$ with pmf $\mathbf{w}_k = (w_{1},\cdots, w_{m})$}
    $\mathbf{x}\leftarrow (\mathbf{x},x_k)$\;
}
$\mathbf{s} \leftarrow $ Trans$(\mathbf{x})$ \tcp*{ Read $\s$ from $\x$.}
\Return{$\mathbf{s}$}
\caption{Sample an $\s$ from $p_n(\s)$ in $O(\text{poly}(n)2^{8n/3})$ time and polynomial space.}
\label{alg:GauSample1}
\end{algorithm}
\DecMargin{0.1cm}
\end{minipage}
\end{center}




\subsection{Step 4: Simplification of marginal probabilities}

In the following, we will show how each of the marginal probability $q_n(x_{1},\cdots,x_{k})$ can be simplified for reducing the computational cost.
The main idea of the simplification is to split the Hafnian of the original large matrix to Hafnians and permanents of some small matrices. Our simplification process can be visualized after defining a set of rules for the graph manipulation.

\subsubsection{Hafnian and graph}

There is a close relation between Hafnian and weighted perfect matchings of graph, which we summarize in Figure~\ref{fig:hafnian_and_graph}. Below, unless otherwise stated, we shall refer to a complete graph as an undirected graph with vertices connected to every other vertices and to itself, i.e. a standard complete graph with \emph{selfloops}. Moreover, since $W_{\x}$ is symmetric, it can be viewed as the adjacency matrix of a complete graph with $n$ vertices $\{ x_i \}$ and weighted edges, where the weight of each edge $(x_i, x_j)$ equals to $W(x_i, x_j)$ (for example, see Eq.~(\ref{eq:W_s_example})). For the sake of illustration, we may use the same notation to denote a graph and its adjacency matrix. 
Note that the weight $W(x_i, x_j)$ here could be a complex number. This is not standard, but helps us visualize our simplification process. Now we are ready to present the first two rules:
\begin{enumerate}
	\item[\textbf{Rule 1}] {The weight of} an edge $(x_i, x_j)$ represents the matrix element $W(x_i, x_j)$, we also use $(x_i,x_j)$ to denote the weight of edge $(x_i, x_j)$ with a little abuse of symbols.
	\item[\textbf{Rule 2}] When multiple disjoint edges are put together, it represents the product of the corresponding matrix elements.
\end{enumerate}
Figure~\ref{fig:hafnian_and_graph}~(b) gives an example of a perfect matching of the vertices $\{ x_i \}$, which is a set of disjoint edges that exactly connects every vertices, and according to Rule 2, can be translated to the product of four matrix elements of $W$.
In this way, each term in the expansion of $\Haf(W_{\x})$ corresponds to a perfect matching of $\{ x_i \}$, and we can represent $\Haf(W_{\x})$ pictorially as Figure~\ref{fig:hafnian_and_graph}~(c).

\begin{figure*}
\center
\includegraphics[trim = 0mm 20mm 0mm 20mm, clip=true,width = 0.9\textwidth]{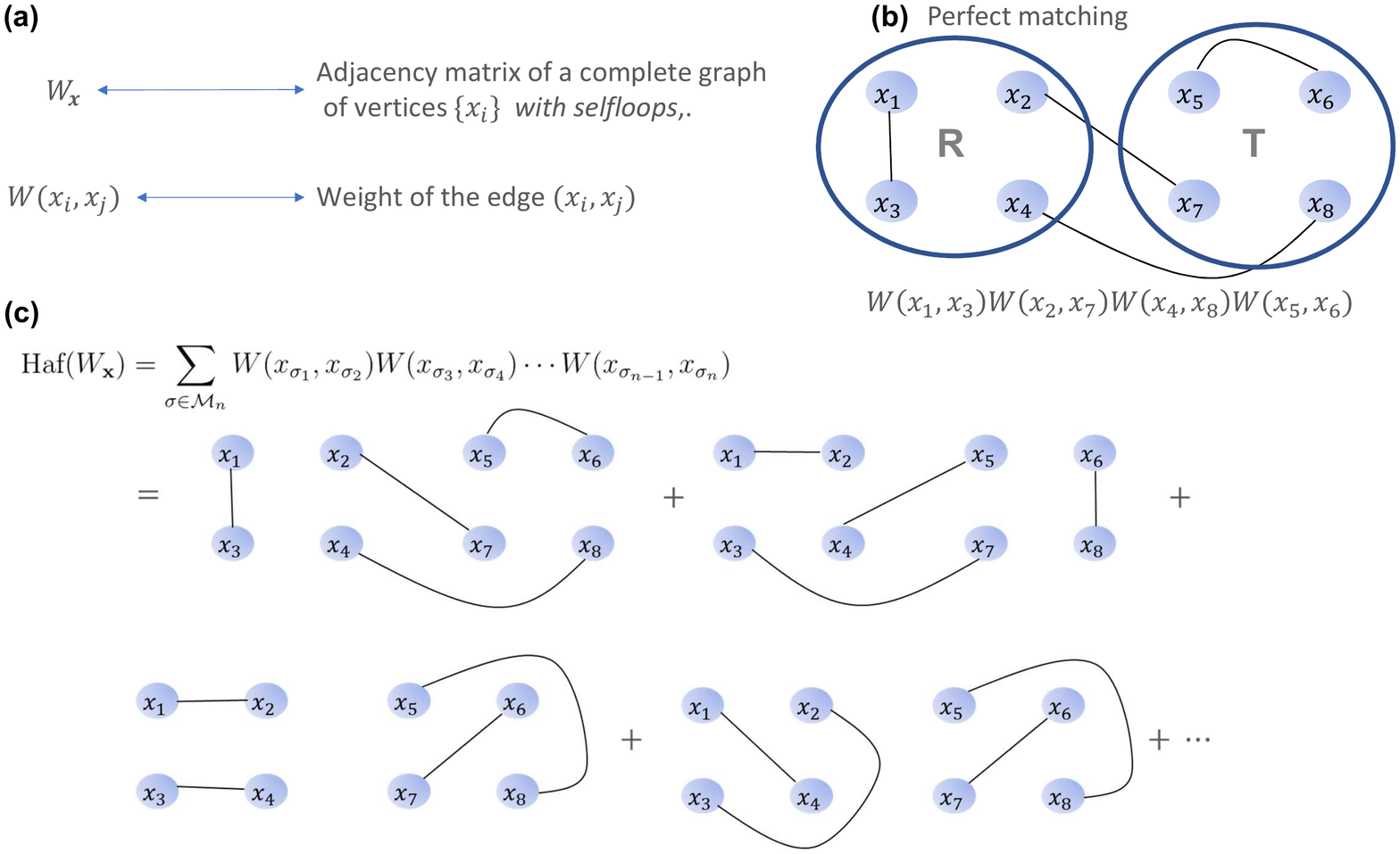}
\caption{Connection between Hafnian and weighted perfect matchings of a graph. \textbf{(a)} $W_{\x}$ can be viewed as an adjacency matrix; \textbf{(b)} a term in the expansion of $\Haf(W_{\x})$ corresponds to a perfect matching of the vertices $\{ x_i \}$; \textbf{(c)} visualization of the expansion of $\Haf(W_{\x})$.
}
\label{fig:hafnian_and_graph}
\end{figure*}

Now for a desired marginal probability $q_n(x_1, \cdots, x_k)$ that we want to calculate, we may partition the vertex set $\{ x_i \}$ into two, $R$ and $T$, where $R$ contains the vertices $\{x_1, \cdots, x_k\}$ and $T$ contains the remaining vertices, so that vertices in $T$ will be summed over.
We can classify the edges into three categories: edges inside $R$, edges inside $T$, and edges across $R$ and $T$. Figure~\ref{fig:hafnian_and_graph}~(b) is an example for one perfect matching in the case $k=4$. In this example, $(x_2, x_7)$ and $(x_4, x_8)$ are interconnected edges while the other two are entirely in $R$ or $T$.

Let $\X_i$ be the set $\{ x_1, x_2, \cdots, x_i \}$.
Let $\a$ be a vector with elements drawing from $\X_k$ without replacement. We use $R_{\a}$ to denote both the complete subgraph with vertices specified by $\a$ and the corresponding adjacency matrix. For example, for $\a = (x_1, x_3)$, $R_{\a}$ is a complele subgraph of vertices $x_1$ and $x_3$, whose adjacency matrix is, 
\ba
R_{\a} = 
\begin{pmatrix}
W(x_1, x_1) & W(x_1, x_3) \\
W(x_3, x_1) & W(x_3, x_3)
\end{pmatrix} \ ,
\ea
as the weight of the edge $(x_i, x_j)$ is given by $W(x_i, x_j)$. We can similarly define $T_{\a'}$, which is a complete subgraph of $T$ with vertices specified by a vector $\a'$, whose elements are in the set $\X_n \backslash \X_k$, i.e. $x_j$ for $j>k$.

Now, let us consider the interconnected edges.
Let $\e$ be a vector with elements drawn from $\X_k$ and $\e'$ be another vector with elements drawn from $\X_n \backslash \X_k$. Then define $G_{\e, \e'}$ to be a complete bipartite subgraph of $G$ with vertices specified by $\e$ and $\e'$ (no selfloops) and we will use the same notation $G_{\e, \e'}$ to denote its biadjacency matrix \footnote{The adjacency matrix of $G_{\e, \e'}$ is $\begin{pmatrix} 0 & G_{\e, \e'} \\ G_{\e, \e'}^T & 0 \end{pmatrix}$}. For example, suppose $\e=(2, 4)$ and $\e'=(7, 8)$, and then
\ba
G_{\e, \e'} = 
\begin{pmatrix}
W(x_2, x_7) & W(x_2, x_8) \\
W(x_4, x_7) & W(x_4, x_8)
\end{pmatrix} \ .
\ea
In the remaining part of this paper, we may use a set as a vector (with indices in an increasing order). For example, the set $\{ x_2, x_1, x_5, x_3 \}$ will be viewed as a vector $(x_1,x_2,x_3,x_5)$.

Now, we are ready to show how a Hafnian of a big matrix can be broken into Hafnians and permanents of small ones. Lemma \ref{lem:ThmSplitHaf} in Appendix \ref{app:HafSplit} gives an intuitive explanation for the splitting process.

\subsubsection{Summation path}

Next, we would present the key idea of our simplification process. 
In the expression of $q_n(x_1, \cdots, x_k)$, we are actually dealing with $|\Haf(W_{\x})|^2$, so we have two graphs to manipulate, $W_{\x}$ and $W^*_{\x}$. Correspondingly, we could define $R^*$ and $T^*$ to be the counterpart of $R$ and $T$, respectively. Recall that variables in $T_i$ and $T_i^*$ are those to be summed over in the expression of $q_n(x_1, \cdots, x_k)$. 

\begin{figure*}
\center
\includegraphics[width = 0.9\textwidth]{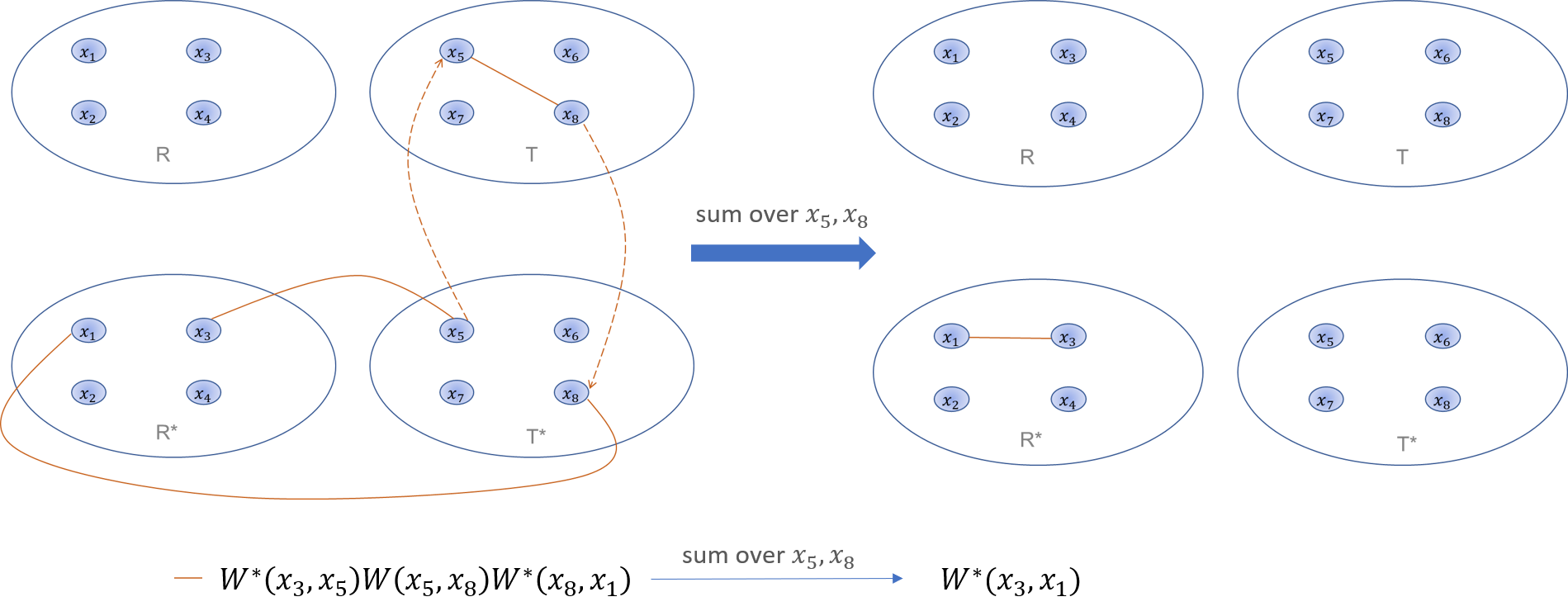}
\caption{Summation induces a path, which can be replaced by an edge connecting the two endpoints.
}
\label{fig:summation_path}
\end{figure*}

By the definition of $W$, we know that it is unitary, which means,
\ba
\sum_{i} W(i, j) W^*(i, j') = \delta(j, j') \ . \label{eq:W_unitary}
\ea
This gives us the third rule:
\begin{enumerate}
	\item[\textbf{Rule 3}] Summation over variables in $T$ and $T^*$ induces a `path', which can be replaced by an edge of the two endpoints.
\end{enumerate}
We shall use a simple example to explain this rule. On the left hand side of Figure~\ref{fig:summation_path}, when the three solid orange edges are put together, it represents $W^*(x_3, x_5) W(x_5, x_8) W^*(x_8, x_1)$, according to Rule 2. Then we sum over $x_5$ and $x_8$, which induces two dashed lines and form a path connecting these four variables (which we call a `summation path'). Using Eq.~\eqref{eq:W_unitary}, we have the following relation,
\ba\label{eq:summation_path_example}
\sum_{x_5, x_8} W^*(x_3, x_5) W(x_5, x_8) W^*(x_8, x_1) = W^*(x_3, x_1) \ .
\ea
The left-hand side represents a summation path with endpoints $x_1$ and $x_3$, and the right-hand side represents an edge $(x_1, x_3)$ in $R^*$, according to Rule 1. This equality means such a summation path is equivalent to an edge.

\begin{figure}[t]
\center
\includegraphics[width = 0.45\columnwidth]{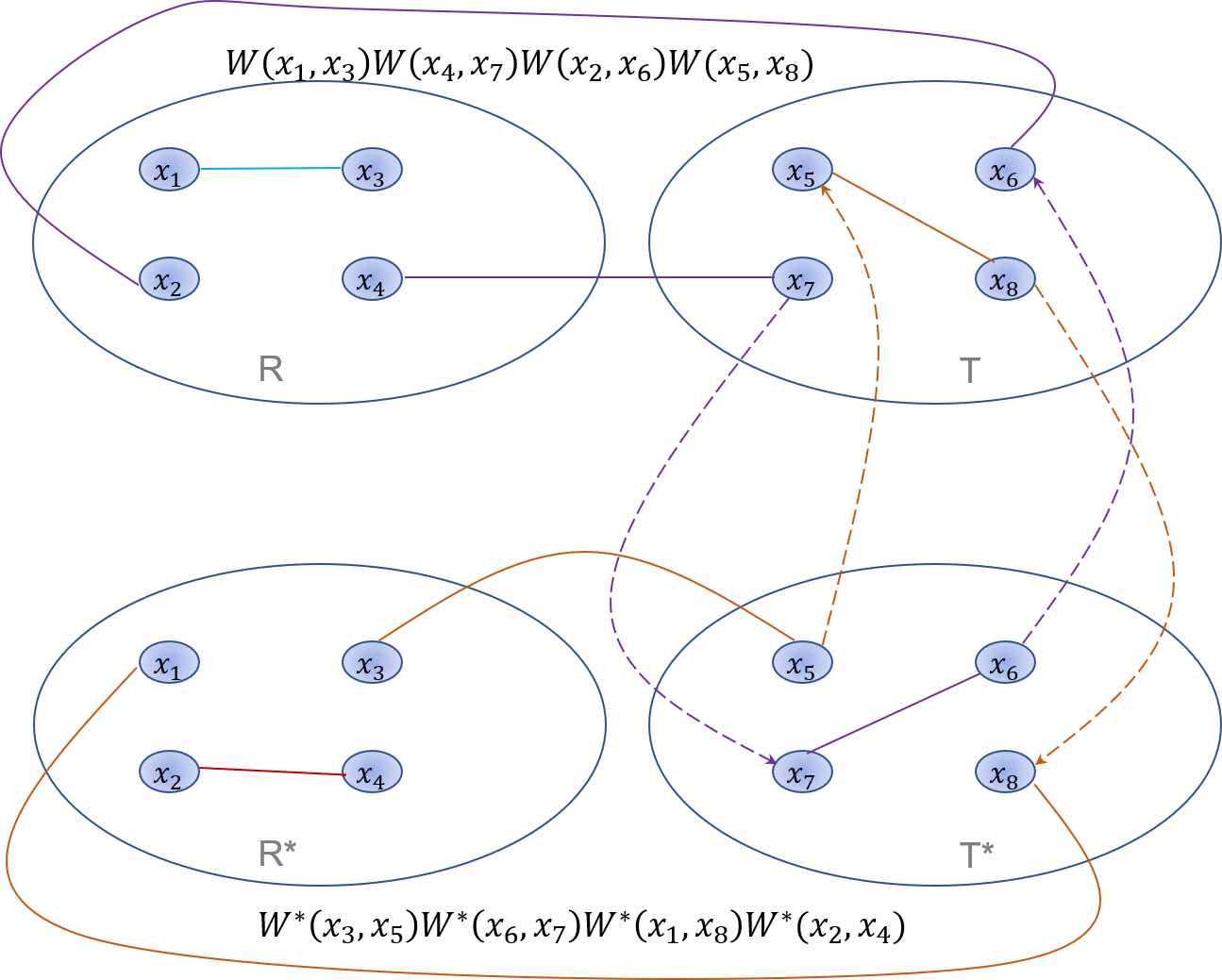}
\caption{One term in the expansion of $q_n(x_1, \cdots, x_k)$.
}
\label{fig:summation_path_example}
\end{figure}

In the expansion of $|\Haf{W_{\x}}|^2$, each term is associated with one perfect matching in $W_{\x}$ and one perfect matching in $W_{\x}^*$. So there are $n$ edges pictorially; see Figure~\ref{fig:summation_path_example} for an example. Applying Rule 3, we can substantially reduce the number of remaining edges, and thus simplify the calculation of the marginal probability. Following a similar calculation to Eq.~\eqref{eq:summation_path_example}, we derive the following lemma.
\begin{lemma}\label{lem:summation_path}
According to the endpoints of the summation path, there are the following four possibilities:
\begin{enumerate}
	\item if both endpoints are in $R$, then applying Rule 3 gives an edge in $R$, which represents $W(x_i, x_j)$ with $i, j \leq k$;
	\item if both endpoints are in $R^*$, then applying Rule 3 gives an edge in $R^*$, which represents $W^*(x_i, x_j)$ with $i, j \leq k$;
	\item if one in $R$ and another in $R^*$, then it gives $\delta(x_i, x_j)$ with $i, j \leq k$;
	\item if no endpoints in $R$ or $R^*$, then it gives a number. 
\end{enumerate}
\end{lemma}
In the expansion of the marginal probability, there are many terms like Figure~\ref{fig:summation_path_example}, and there are two layers of summation (Eq.~\eqref{eq:marginal}). One is the summation over variables in $T$ ($T^*$), which induces summation path and leads to some level of simplification. The other is the summation over perfect matchings of $W_{\x}$ and $W_{\x}^*$, which allows us to group the remaining terms after applying Rule 3. This eventually gives us an theorem, stating that the marginal probability can be decomposed into summation over small Hafnians and permanents.

\begin{theorem}
The marginal distribution of the sequence $(x_{1},\cdots,x_{k})$ from the distribution $q_n(\x)$ can be expressed as follows,

\begin{align}
q_n(x_1, \cdots, x_k) =
\sum_{\substack{
j_1, j_2\\
\a, \a'
} }
\Haf(R_{\a}) \Haf(R^*_{\a'}) \sum_{\substack{\mu \in S_{\mu}  }} F_{\mathbf{j}}^{(\mu)}
 \sum_{\substack{
 A, B\\
 \e, \e'
 }} \Per(S_{A,B}) \Haf(R_{\e}) \Haf(R^*_{\e'})  \ , \label{eq:marginal_final_expression}
\end{align}
where $F_{\mathbf{j}}^{(\mu)} :=\frac{1}{f_n} F(k, \mu, j_1, j_2) =\frac{(n-k)!}{f_n} \binom{\frac{n-k+\mu+m}{2}-1}{k-j_{1}-j_{2}+\frac{m}{2}-1}$. The ranges of the summation variables are as follows: 1) $j_1$ and $j_2$ are both integers from $\max(0, k - \frac{n}{2})$ to $\frac{k}{2}$; 2) $\vb{a} \in \binom{X_k}{2j_1}$ and $\a' \in \binom{X_k}{2j_2}$; 3) $S_{\mu} \equiv \{\mu \in \mathbb{N}: (3k-2(j_{1}+j_{2})-n) \leq \mu \leq k - 2\max(j_{1},j_{2}), k \equiv \mu \text{ mod 2} \}$ is the range of $\mu$; 4) $A \in \binom{X_k\backslash \a}{\mu}$ and $B \in \binom{X_k\backslash \a}{\mu}$; 5) $\e = X_k\backslash\{ \a\cup A \}$ and $\e' = X_k\backslash\{ \a'\cup B \}$.
$S_{A,B}$ is defined as $S_{A,B}(i, j) \equiv \delta({A_i}, {B_j})$.
\label{thm:marginal}
\end{theorem}

In this theorem, 
$\binom{S}{i}$ denotes the set of all possible combinations of $i$ elements from the set $S$. As an example, $\binom{\X_3}{2} = \{ (x_1,x_2), (x_1,x_3), (x_2,x_3) \}$.
$\Haf(R_{\a})$ and $\Haf(R_{\a'}^*)$ are formed by those edges not involving in the summation path. $\Haf(R_{\e})$ is from case 1 in Lemma~\ref{lem:summation_path} and $\Haf(R_{\e'}^*)$ is from case 2. $\Per(S_{A,B})$ is from case 3 and $F(k, \mu, j_1, j_2)$ is the summation of all numbers simplified from case 1, 2, 3 and 4. However, the proof is actually rather involved, and we leave it in Appendix~\ref{app:proof_theorem_1}.

If we restrict to the collision-free regime, i.e., the regime that $x_i \neq x_j$, then the expression of $q_n(x_1, \ldots, x_k)$ can be further simplified. In Theorem~\ref{thm:marginal}, by the definition of $S_{A, B}$, if there are one element in $A$ different from all elements in $B$, then one row in $S_{A, B}$ will be all zero, which means $\Per(S_{A, B})$ will also be zero. Now in the collision-free regime, in order for $S_{A, B}$ not to have an all-zero row, it should be that for every element $x_i \in A$, there is also an $x_i \in B$. That is, $\Per(S_{A,B}) \neq 0$ if $A = B$. Furthermore, there is at most one $1$ in each row in this regime, which implies the permanent is either 0 or 1. So we have the following proposition. 
\begin{proposition}
In the collision-free regime, 
\ba
\Per(S_{A, B}) = 
\begin{cases}
1, & \text{if } A = B \\
0, & \text{otherwise}
\end{cases} \ .
\ea
\end{proposition}

\subsection{Numerical results}\label{subsec:numerical}

\begin{figure*}
\center
\includegraphics[trim = 0mm 45mm 0mm 48mm, clip=true,width = 1.0\textwidth]{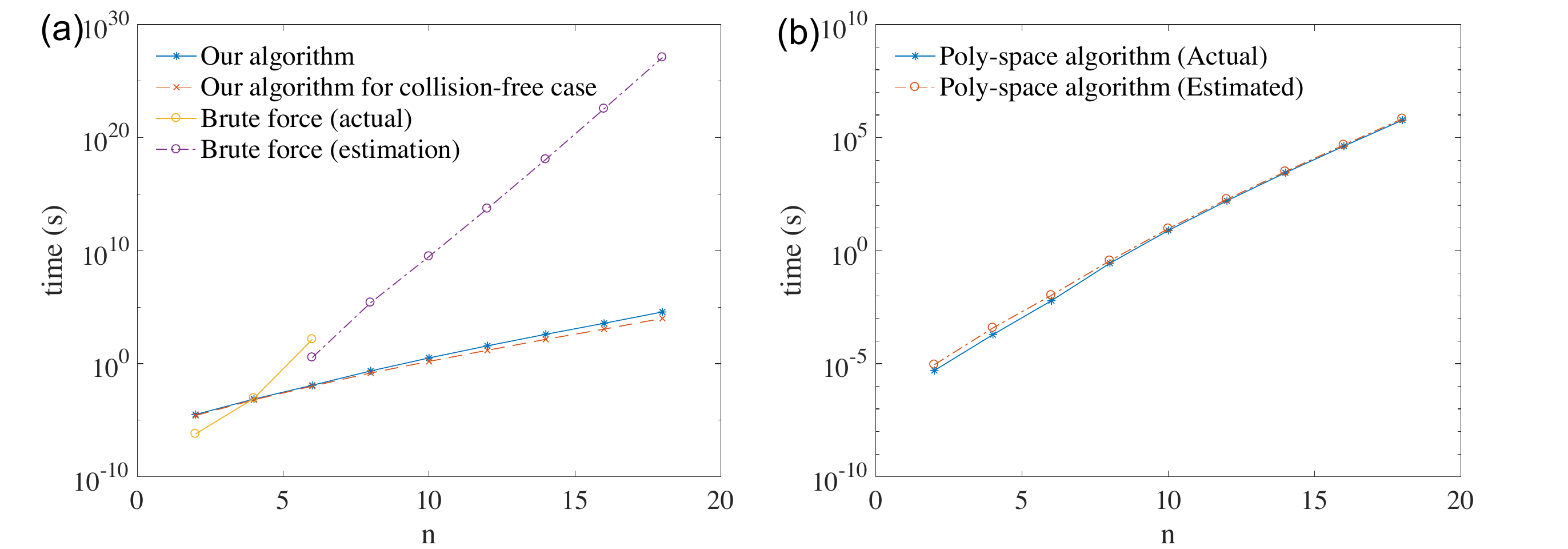}
\caption{(a) Comparison of the running time of brute force sampling and our exponential space algorithm (both full space and collision-free space).
(b) Comparison of the actual running time with the estimated one.
}
\label{fig:comparison}
\end{figure*}

In the numerical simulation, the number of modes is set to be the square of photon number, that is $m = n^2$. 
Using our algorithm, we can sample 18 photons in about 20 hours on a laptop. As a comparison, we can only sample 6 photons on the same laptop using the brute-force sampling (details about the brute-force sampling can be found in Appendix~\ref{app:brute-force}). Figure~\ref{fig:comparison} (a) shows the comparison of the running time (in log scale) of brute force sampling and of our polynomial-space algorithm. The data points for actual running time of our algorithm for $n \leq 14$ and of brute-force sampling are from 300 repetitions. For $n > 6$, we give an estimation of the running time of the brute-force sampling, which is much more slower than that of our algorithm. 
Figure~\ref{fig:comparison} (b) presents the actual running time of the exponential-space algorithm (blue solid lines), which is faster than the polynomial-space one as expected.

To check that our algorithm output the right results, we compared the distribution generated by our algorithm and the distribution from brute-force sampling for $n = 4$ and $m = 16$, as shown in Figure \ref{Comparision}. To restore the original distribution for our sampling algorithm as much as possible, we sample 400,000 times independently.
The horizontal axis is $-\log{p}$ and the vertical axis is $\Pr(-\log{p})$. From Figure \ref{Comparision} (a), we find the distributions by the two sampling algorithms are both close to theoretical value when the probability $p$ is not too small. On the other hand, when $p$ is small, the distribution value by our our algorithm and by Brute Force sampling are consistent with each other. In Figure \ref{Comparision} (b), we list the comparison of distribution of our algorithm and theoretical value. 

\begin{figure}[t!]
\includegraphics[trim = 0mm 45mm 0mm 48mm, clip=true, width = 1.0\textwidth]{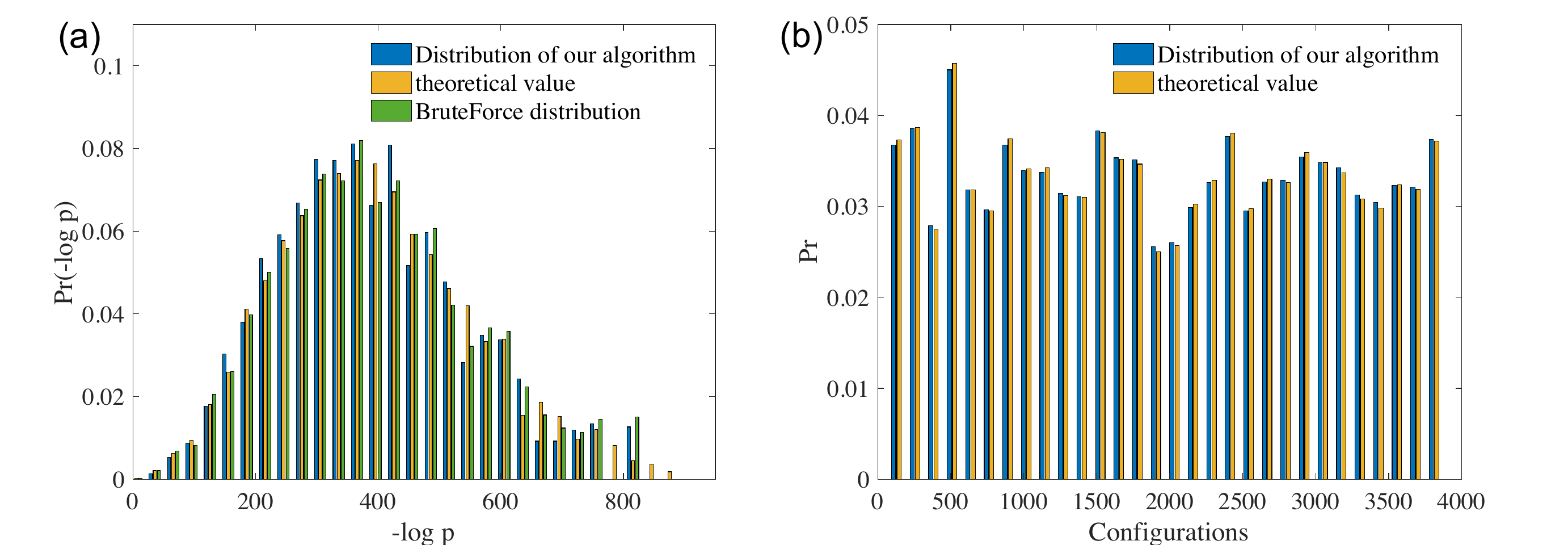}
\caption{(a) Comparision of distribution from our classical simulation algorithm (blue), brute-force sampling (green) and theorical calculation (yellow) when $n = 4$ and $m = 16$, when $p$ is not too small.  (b) Similar comparison when $p$ is small.}
\label{Comparision}
\end{figure}

Together with the first part of our classical sampling process (Sample a photon number $n$, in which we set $N = 50n_\text{most}$.), we give the frequencies of each output mode for photon number $>1$ and $= 1$, in which $m = 36$ in Figure \ref{fig:photonCom}.
From Figure \ref{fig:photonCom} we know when $r$ is small enough, the settings can be restricted to collision-free. 
\begin{figure*}
    \centering
\includegraphics[trim = 0mm 45mm 0mm 48mm, clip=true,width = 1.0\textwidth]{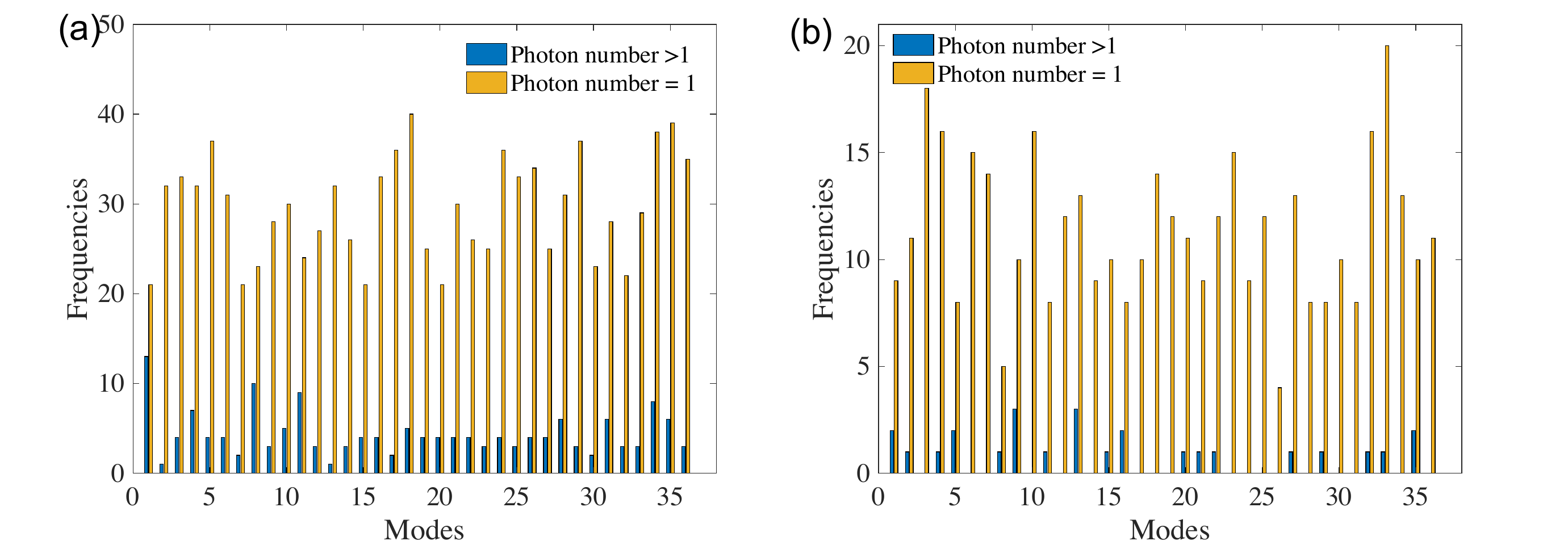}
    \caption{Comparison of frequencies for different photon numbers of each mode after independently sampling 300 times, in which $m = 36$, and (a) $r = 0.3423$, (b) $r = 0.2$.}
    \label{fig:photonCom}
\end{figure*}

To estimate the running time on high-performance computers, we first benchmark the running time of calculating Hafnians and permanents on a laptop, and then estimate the time directly from Theorem \ref{thm:marginal}. The results are shown in Figure~\ref{fig:comparison}~(b) (yellow dashed line).
We can see that the estimated time for 2 $\sim$ 18 photons is consistent with the actual running time on our laptop, justifying our estimation. 
After that, we can transform the estimated time into the required basic operations and results for photon number 20 $\sim$ 30 are presented in Table \ref{EstimateTable}.



We only consider the polynomial-space algorithm here, since the exponential-space algorithm requires massive inter-process communication, thus increasing time cost significantly. Besides, the polynomial-space algorithm allows parallel implementation.
From this estimation, we predict that Sunway TaihuLight \cite{fu2016sunway}, which can implement $10^{17}$ floating-point operations per second, can sample about 30 photons within one day.

To check the correctness of our analysis for large-scale parallel computation in Sunway TaihuLight, we test our algorithm on the HuaWei KunLun server for sampling a configuration for 20 photons, which use 256 cores and about $4.2\times 10^7$ CPU times. The actual time is about 36 hours, with average parallelism 325. The CPU time on the HuaWei Kunlun server is 4 times longer than our estimation since some extra operations for parallel process are required. However, this should not affect our estimated time for Sunway TaihuLight in the order of magnitude.

\begin{table*}
\centering
\begin{tabular}{c|c|c|c}
\hline\hline
    {Photons}  & 20 & 22 & 24 \\
    \hline
    {Estimated non-parallel time ($s$)}  & $9.37\cdot 10^7$ &$1.25\cdot 10^8$ & $1.66\cdot 10^9 $\\
    \hline
    {Estimated basic operations} &  $4.28\cdot 10^{15} $&$5.74\cdot 10^{16}$ & $7.59\cdot 10^{17}$ \\
    \hline
    Estimated paralleled time on Sunway TaihuLight (s) & 0.0428 & 0.574 & 7.59\\
      \hline\hline
     Photons & 26 & 28 & 30\\
    \hline
  Estimated non-parallel time 
   (s)& $2.18\cdot 10^{10}$ & $2.85 \cdot 10^{11}$ & $3.76\cdot 10^{12}$\\
    \hline
  Estimate of basic operations  & $9.96 \cdot 10^{18}$ & $1.30\cdot 10^{20}$ & $1.72 \cdot 10^{21}$\\
  \hline
 Estimated paralleled time on Sunway TaihuLight (s)& 99.6 &$ 1.30 \cdot 10^{3} $& $1.72 \cdot 10^{4}$ \\
    \hline\hline
\end{tabular}
\caption{Estimation of times and basic operations for even photons ranging from $20$ to $30$ by Alg. \ref{alg:GauSample1} with polynomial space.}
\label{EstimateTable}
\end{table*}

\section{Conclusion\label{conclusion}}

In this paper, we study the Gaussian Boson Sampling problem and give a polynomial-space classical algorithm with time complexity $O(m\sinh^{2} r) + O(\text{poly}(n) 2^{8n/3})$, which is far more efficient than the brute-force sampling method. The time complexity of our algorithm can be improved to $O(m\sinh^{2} r) + O(\text{poly}(n) 2^{5n/2})$ if exponential space is used to store intermediate calculation results. Nevertheless, our numerical results implies the above two bounds are far from tight. Appendix~\ref{sec:more_numerical} shows a comparison of actual executed time with the theoretical bounds

We benchmark our algorithm on a laptop and on Huawei Kunlun server. The former can sample 18 photons in 20 hours while the latter can sample 20 photons in 36 hours. Based on our algorithm, Sunway TaihuLight is estimated to be able to sample about 30 photons. These numbers are smaller than that of standard boson sampling, which suggests that GBS may be more feasible for demonstrating quantum computational supremacy.

Note: Recently, Quesada et al.~\cite{quesada2019classical} independently presented a classical simulation algorithm for general GBS, where they can simulate about 14 photons and 100 modes in about $10^3$ seconds.


\bibliographystyle{plain}
\bibliography{ref}


\newpage

\begin{appendix}

\section{\label{Normalization}Proof of the normalization of $q_n(x)$}

\begin{proof}
By the definition of Hafnian, we can let $\sigma_{1} = n$ in Equation \eqref{HafMatch} and thus
\begin{align*}
\text{Haf}(V) =\sum_{u\in [n-1]} V_{n,u}
\sum_{\mbox{\tiny$\begin{array}{c}
\sigma\in \mathcal{M}_{[n-1]\backslash \{u\}}
\end{array}$}}
\prod_{j=1}^{n/2-1}V_{\sigma_{2j-1},\sigma_{2j}}
\end{align*}
where $\mathcal{M}_{[n-1]\backslash \{u\}}$ is the collection of all of perfect matchings of $[n-1]\backslash \{u\}$. Let $W_{x}^{S}$ denote matrix obtained by selecting rows and columns of set $x_{S}$ of W simultaneously. Let $f_{n} = \sum_{x\in[m]^{n}}|\text{Haf}(W_{\x}^{[n]})|^{2}$, then
\begin{align}
f_{n}=&\sum_{x\in[m]^{n}}|\text{Haf}(W_{\x})|^{2} \\
=&\sum_{x\in[m]^{n}}\sum_{u,v\in [n-1]} W_{x_{n},x_{u}}W_{x_{n},x_{v}}^{*}
\sum_{\mbox{\tiny$\begin{array}{c}
\sigma\in M_{[n-1]\backslash \{u\}}\\
\tau\in M_{[n-1]\backslash \{v\}}
\end{array}$}}
\prod_{j=1}^{n/2-1}W_{x_{\sigma_{2j-1}},x_{\sigma_{2j}}} W_{x_{\tau_{2j-1}},x_{\tau_{2j}}}^{*}\\
&=\sum_{\mbox{\tiny$\begin{array}{c}
x\in[m]^n, u = v\\
\sigma\in M_{[n-1]\backslash \{u\}}\\
\tau\in M_{[n-1]\backslash \{u\}}
\end{array}$}}\prod_{j = 1}^{n/2 - 1}W_{x_{\sigma_{2j-1}},x_{\sigma_{2j}}}W_{x_{\tau_{2j-1}},x_{\tau_{2j}}}^{*} + \sum_{\mbox{\tiny$\begin{array}{c}
u\ne v\in[n-1]\\
\sigma\in M_{[n-1]\backslash \{u\}}\\
\tau\in M_{[n-1]\backslash \{u\}}
\end{array}$}}\prod_{j = 1}^{n/2 - 1}W_{x_{\sigma_{2j-1}},x_{\sigma_{2j}}}W_{x_{\tau_{2j-1}},x_{\tau_{2j}}}^{*}\\
&=(n-1)mf_{n-2}+(n-1)(n-2)f_{n-2} \label{Eqf}\\
&=(n-1)(m+n-2)f_{n-2}\\
&=\binom{\frac{m+n}{2}-1}{\frac{n}{2}}n!
\end{align}
Equation \eqref{Eqf} holds by the fact that $W$ is unitary and symmetric, \emph{i.e.},
\begin{align*}
\sum_{x_{n}\in[m]} &W_{x_{n},x_{u}}W_{x_{n},x_{v}}^{*} = [x_{u}=x_{v}]=[u = v] + [u\ne v,x_{u} = x_{v}].
\end{align*}
\end{proof}

\section{Equivalence of $W_{\s}$ and $W_{\x}$}
\label{App:equivalence}

Plug $W_{\x}(i, j) = W(x_i, x_j)$ into the expression of $\Haf(W_{\x})$, which gives
\ba
\Haf(W_{\x}) &=& \sum_{\sigma \in \cal{M}_n} W_{\x}(\sigma_1, \sigma_2) W_{\x}(\sigma_3, \sigma_4) \ldots W_{\x}(\sigma_{n-1}, \sigma_n) \\
&=& \sum_{\sigma \in \cal{M}_n} W(x_{\sigma_1}, x_{\sigma_2}) W(x_{\sigma_3}, x_{\sigma_4}) \ldots W(x_{\sigma_{n-1}}, x_{\sigma_n}) \ .
\ea
If $\x'$ is obtained from $\x$ through a permutation $\tau$, then $\x' = (x_{\tau_1}, x_{\tau_2}, \ldots, x_{\tau_n})$. So for a perfect matching $\sigma$, $x'_{\sigma_i} = x_{\sigma'_i}$, where $\sigma'$ is another perfect matching. Then
\ba
\Haf(W_{\x'}) &=& \sum_{\sigma \in \cal{M}_n} W(x'_{\sigma_1}, x'_{\sigma_2}) W(x'_{\sigma_3}, x'_{\sigma_4}) \ldots W(x'_{\sigma_{n-1}}, x'_{\sigma_n}) \\
&=& \sum_{\sigma' \in \cal{M}_n} W(x_{\sigma'_1}, x_{\sigma'_2}) W(x_{\sigma'_3}, x_{\sigma'_4}) \ldots W(x_{\sigma'_{n-1}}, x_{\sigma'_n}) \ ,
\ea
Thus $\Haf(W_{\x'}) = \Haf(W_{\x})$ if $\x'$ is obtained from $\x$ through permutation.

Now, given an $\x$, we can read an $\s$. We will show that $\Haf(W_{\x}) = \Haf(W_{\s})$. First, we sort $\x$ with an increase order to get $\z$. For example, if $\x = (4, 2, 4, 1)$, then after an increase sort, we will get $\z = (1, 2, 4, 4)$. Increase sort is actually a permutation, so we have $\Haf(W_{\x}) = \Haf(W_{\z})$. Moreover, $\Haf(W_{\z}) = \Haf(W_{\s})$, which can be verified from the definition of $W_{\s}$. Thus, $\Haf(W_{\x}) = \Haf(W_{\s})$.

\section{\label{timeA}Time complexity when using polynomial space.}
By \cite{bjorklund2018faster}, the time complexity of Hafnian of a $n\times n$ matrix is $O(n^32^{n/2})$, and the time complexity of $\Per(S_{A,B})$ is $O(\mu)$.

The time complexity of Alg. \ref{alg:GauSample1} is bounded by
\begin{align}
&cm\sum_{k = 0}^{n}\sum_{ja = \max( 0, 2k - n)}^{k}\sum_{jb = \max( 0, 2k - n)}^{ja} \binom{k}{ja}\binom{k}{jb}ja^3 jb^3 2^{(ja+jb)/2}\sum_{ \mu = \max(0,3k-ja-jb-n)}^{k - ja}   \\ 
&\binom{k-ja}{\mu}\binom{k-jb}{\mu}
(k-ja -\mu)^3 (k - jb - \mu)^3
2^{(k-ja-\mu+k-jb-\mu)/2} \mu \\ 
&\leq  \text{poly}(n)\sum_{k = 0}^{n/2} \sum_{ja = 0}^{k}\sum_{jb = 0}^{k} \binom{k}{ja}\binom{k}{jb}2^{k}\sum_{0\leq \mu \leq k - ja}\binom{k-ja}{\mu}\binom{k-jb}{\mu}    \\
& + \text{poly}(n)\sum_{k = n/2}^{n} \sum_{ja = 2k - n}^{k}\sum_{jb = 2k -n}^{k} \binom{k}{ja}\binom{k}{jb}2^{k}\sum_{0\leq \mu \leq k - ja}\binom{k-ja}{\mu}\binom{k-jb}{\mu}    \\
&\leq \text{poly}(n) \sum_{k = 0}^{n/2} 2^{k}\sum_{\mbox{\tiny$\begin{array}{c}
0\leq ja \leq k\\
0\leq jb \leq k
\end{array}$}} \binom{k}{ja}\binom{k}{jb}\binom{2k-ja-jb}{k-jb}+ \text{poly}(n)\sum_{k = 0}^{n/2}2^{n-k}\sum_{\mbox{\tiny$\begin{array}{c}
0\leq ja \leq k\\
0\leq jb \leq k
\end{array}$}}\binom{n-k}{ja}\binom{n-k}{jb}\binom{ja+jb}{jb}\\
&\leq \text{poly}(n) \sum_{k = 0}^{n/2} 2^{k}
\sum_{\mbox{\tiny$\begin{array}{c}
0\leq ja \leq k\\
0\leq jb \leq k
\end{array}$}}\binom{k}{ja}\binom{k}{jb}\binom{2k-ja}{k-jb} + \text{poly}(n)\sum_{k = 0}^{n/2}2^{n-k}\sum_{\mbox{\tiny$\begin{array}{c}
0\leq ja \leq k\\
0\leq jb \leq k
\end{array}$}}\binom{n-k}{ja}\binom{n-k}{jb}\binom{ja+jb}{jb}   \\
&\leq \text{poly}(n)\sum_{k = 0}^{n/2} 2^{k}\sum_{ja = 0}^{k}\binom{k}{ja}\binom{3k-ja}{k} + \text{poly}(n)\sum_{k = 0}^{n/2}k^{2}2^{8n/3}\label{EqNotEq}\\
&\leq \text{poly}(n)\sum_{k = 0}^{n/2} 2^{k} \binom{4k}{2k} + \text{poly}(n) n^{3} 2^{8n/3}   \\
&\leq \text{poly}(n) 2^{8/3n}  
\end{align}
for some constant $c$. Eq. \eqref{EqNotEq} holds by the result of Lemma \ref{LemEq}.
\begin{lemma}
\begin{equation}
2^{n-k}\binom{n-k}{x}\binom{n-k}{y}\binom{x+y}{x} \leq 2^{8n/3}
\label{EqExpNotEq}
\end{equation}
\label{LemEq}
\end{lemma}
in which $0\leq k \leq n/2$ and $0\leq x, y \leq k$.

\begin{proof}
If $k\geq n/3$, then $\frac{n-k}{2}\leq k$. Thus
\begin{align*}
2^{n-k}\binom{n-k}{x}\binom{n-k}{y}\binom{x+y}{x}&\leq 2^{n-k}\binom{n-k}{\frac{n-k}{2}}^{2}\binom{2k}{k}\leq 2^{n-k}2^{2(n-k)}2^{2k}\leq 2^{8/3n}
\end{align*}
If $k< n/3$, then
\begin{align*}
2^{n-k}\binom{n-k}{x}\binom{n-k}{y} \binom{x+y}{y} \leq 2^{n-k}\binom{n-k}{k}^{2}\binom{2k}{k}\leq 2^{n+k}\binom{n-k}{k}^{2}
\end{align*}
Let $f(k) = 2^{n+k}\binom{n-k}{k}^{2}$, it is easy to find
$$\frac{f(k)}{f(k-1)} = \frac{2(n-2k+1)(n-2k+2)}{k(n-k+1)}> 1$$
when $k\leq n/3$, thus $f(k) < f(n/3)$ when $k < n/3$.
\end{proof}

 \section{\label{timeB}Time complexity when using $O(m2^{n})$ space.}
 Observe that if we conserve Haf$(W_{\mathbf{x}}^{s})$ for all of $s$, we will avoid a lot of repeated computation, and use additional $O(2^{n})$ space.
 The space complexity is,
\begin{align*}
\sum_{k = 0}^{n}\sum_{j=0}^{\frac{k}{2}}\binom{k}{2j}=O(m2^{n})
\end{align*}
There need another $ m\sum_{k=1}^{n} \sum_{j=1}^{\lfloor k/2\rfloor} \binom{k}{2j}T_{2j}=O(mn^{3}(1+\sqrt{2})^{n})$ 
time to compute all of Hafnians of the sub-matrix, where $T_{2j} = O(j^3 2^j)$.
Similarly with the above analysis, we find that when save all of the Hafnians, and the time complexity for sampling process is
\begin{align*}
&\leq cm\sum_{k = 0}^{n}\sum_{ja = \max( 0, 2k - n)}^{k}\sum_{jb = \max( 0, 2k - n)}^{ja} \binom{k}{ja}\binom{k}{jb}\sum_{ \mu = \max(0,3k-ja-jb-n)}^{k - ja} \binom{k-ja}{\mu}\binom{k-jb}{\mu}2^{k-\mu - \frac{ja + jb}{2}}\\
&\leq cm\sum_{k = 0}^{n/2} 2^k\sum_{ja = 0}^{k}\sum_{jb = 0}^{k} \binom{k}{ja}\binom{k}{jb}\sum_{0\leq \mu \leq k - ja}\binom{k-ja}{\mu}\binom{k-jb}{\mu}\\
& +cm \sum_{k = n/2}^{n} 2^{n-k}\sum_{ja = 2k - n}^{k}\sum_{jb = 2k -n}^{k} \binom{k}{ja}\binom{k}{jb}\sum_{0\leq \mu \leq k - ja}\binom{k-ja}{\mu}\binom{k-jb}{\mu}\\
&\leq  cm\sum_{k = 0}^{n/2} 2^{k}\sum_{ja = 0}^{k}\sum_{jb = 0}^{k} \binom{k}{ja}\binom{k}{jb}\binom{2k-ja-jb}{k-jb}+ cm\sum_{k = 0}^{n/2}2^{k}\sum_{ja=0}^{k}\sum_{jb = 0}^{k}\binom{n-k}{ja}\binom{n-k}{jb}\binom{ja+jb}{jb}\\
&\leq cm\sum_{k = 0}^{n/2}2^{k} \binom{4k}{2k} + cm\sum_{k = 0}^{n/2}k^2\cdot 2^k \binom{2n}{n}\\
&\leq cmn^3 2^{5n/2}
\end{align*}

\section{\label{AlgLemma}Pseudocode code to compute weight of $l$ for $x_k$, where $1\leq l\leq k$. }
Alg. \ref{GauSample2} is the pseudocode code to compute the value proportional to probability $q(x_1,\cdots, x_k)$, where $x_1,\cdots, x_{k-1}$ is already been sampled, and $1 \leq x_k  \leq l$.
\begin{center}
\begin{minipage}{1.0\linewidth}
\IncMargin{0.1cm}
\begin{algorithm}[H]
\SetKwInOut{Input}{input}
\SetKwInOut{Output}{output}
\Input{$m,n$ and $m\times m$ matrix $U$ and $x_{1},\cdots, x_{k}$.}
\Output{The value proportional to probability $q(x_{1},\cdots, x_{k})$.}
\emph{$sum\leftarrow 0$}\;
\emph{Compute and save the values of $\binom{\frac{n-k + \mu + m}{2} - 1}{k + \frac{m - ja - jb}{2} - 1}$ for all of $ja, jb, \mu$ in polynomial time and space}\;
\For{$ja \leftarrow \max(0,2k - n)$ to $k$ step by 2}
{
    \For{$jb \leftarrow \max(0, 2k - n)$ to $ja$ step by 2}
    {
        $a\leftarrow \max(0,3k - ja - jb - n), b \leftarrow k - \max(ja, jb)$\;
        \For{$s\in\binom{[k]}{ja}, s' \in\binom{[k]}{jb}$}
        {
            sum$_{1} \leftarrow 0$\;
            \For{$\mu \leftarrow a$ to $b$ step by 2}
            {
                 {temp $\leftarrow 0$}\;
                 \For{$A\in \binom{[k]\backslash s}{\mu}, B\in \binom{[k]\backslash s'}{\mu}$}
                 {
                     $e \leftarrow [k]\backslash \{s \cup A\}, e' \leftarrow [k]\backslash \{s \cup B\}$\;
                     \emph{Construct the 0-1 matrix $S_{A,B}$ with $S_{A,B}(i,j)$ equals 1 if and only if $x_{A_{i}} = x_{B_{j}}$}\;
                     temp$\leftarrow$ temp$ + $Haf$(W_{x}^{e})$ Haf$(W_{x}^{*e'})$Per$(S_{A,B})$\;
                 }
                 sum$_{1} \leftarrow $sum$_{1} + \binom{\frac{n-k + \mu + m}{2} - 1}{k + \frac{m - ja - jb}{2} - 1} \cdot $temp\;
            }
            sum$ \leftarrow $sum$ + $ Haf$(W_{x}^{s})$ Haf$(W_{x}^{*s'}) \cdot \text{sum}_{1}$\;
        } 
        \If{$jb < ja$}{
            sum $\leftarrow $ 2 * real(sum)\tcp*{real($z$) return the real part of complex $z$}
        }
   }
}
\Return{sum}
\caption{Pseudocode code for computing $q(x_{1},\cdots, x_{k})$.}
\label{GauSample2}
\end{algorithm}
\DecMargin{0.1cm}
\end{minipage}
\end{center}

\section{Compute Hafnain by splitting technique.\label{app:HafSplit}}

The following lemma gives some intuitive explanation for Theorem \ref{thm:marginal}.
\begin{lemma}
\begin{align}
    \Haf(W_{\x}) = \sum_{ j }  \sum_{\mathbf{a}, \a', \e, \e' } \Haf(R_{\mathbf{a}}) \Haf(T_{\mathbf{a}'}) \Per(G_{\mathbf{e}, \mathbf{e'}}) \ ,
    \label{eq:splitUp}
\end{align}
where the ranges of the summation variables are as follows: 1) $j$ is an integer from $\max(0, k- \frac{n}{2})$ to $\frac{k}{2}$; 2) $\mathbf{a} \in\binom{\X_k}{2j}$; 3) $\mathbf{a}' \in \binom{\X_n\backslash \X_k}{n-2k-2j}$; 4) $\mathbf{e} = \X_k\backslash \mathbf{a}$; 5) $\mathbf{e}'=\X_n\backslash ( \X_k\cup \mathbf{a}')$. 
\label{lem:ThmSplitHaf}
\end{lemma}


\begin{proof}
In the following, we use quadruple $(\mathbf{a}, \mathbf{a'},\mathbf{e}, \mathbf{e'})$ denote four vertex sets of $G$, where  $\mathbf{a} \in \binom{X_k}{2j}, \mathbf{e} = X_k \backslash \mathbf{a}, \mathbf{a}' \in \binom{X_n\backslash X_k}{n-2k-2j}, \mathbf{e}' = X_n\backslash (X_k \cup \mathbf{a}')$, which construct subgraphs triple $(R_{\mathbf{a}}, T_{\mathbf{a}'}, G_{\mathbf{e}, \mathbf{e'}})$ of $G$, where $R_{\mathbf{a}}, T_{\mathbf{a}'}$ are two complete graphs, and $G_{\mathbf{e}, \mathbf{e'}}$ is a bipartite graph, edges only exists between vertex sets $\mathbf{e}$ and $\mathbf{e}'$. 
We need to prove that any two perfect matchings of two distinct $(R_{\mathbf{a}}, T_{\mathbf{a}'}, G_{\mathbf{e}, \mathbf{e'}})$ triples represent different perfect matching in $G$, and any perfect matching of $G$, which is also a perfect matching of triple $(R_{\mathbf{a}}, T_{\mathbf{a}'}, G_{\mathbf{e}, \mathbf{e'}})$.

\begin{itemize}
    \item [1)]For any two different vertex sets $(\mathbf{a}, \mathbf{a'},\mathbf{e}, \mathbf{e'})$ and $(\theta, \theta',\mathbf{b}, \mathbf{b'})$, suppose $\mathbf{a}\ne \theta$, then there exists a vertex $u$ in $\mathbf{a}$ and $\mathbf{b}$ simultaneously (Since $(\theta,\mathbf{b})$ is a partition of $[k]$). Suppose $M_1,M_2$ are two perfect matchings of $(R_{\mathbf{a}}, T_{\mathbf{a}'}, G_{\mathbf{e}, \mathbf{e'}})$ and $(R_{\theta}, T_{\theta'}, G_{\mathbf{b}, \mathbf{b'}})$ repectively. Suppose matching pair $(u, v)$ is in subgraph $R_{\mathbf{a}}$, and matching pair $(u,v')$ is in subgraph $G_{\mathbf{b}, \mathbf{b'}}$, then $v \in \mathbf{a}$ are different from $v' \in \mathbf{b}'$. Thus $M_1$ and $M_2$ are two different subgraphs in $G$.
    \item [2)] Suppose $M$ is a perfect matching in $G$. In the following we prove that there exists a quadruple $(\mathbf{a}, \mathbf{a'},\mathbf{e}, \mathbf{e'})$ and a perfect matching triple $(M_1, M_2, M_3)$ of their representing triple $(R_{\mathbf{a}}, T_{\mathbf{a}'}, G_{\mathbf{e}, \mathbf{e'}})$, such that $(M_1, M_2, M_3)$ construct $M$.
    As in figure \ref{fig:matchings}, we partition the vertex in $G$ into two sets $R$ and $T$. For a matching pair $(u,v)$, we construct quadruple $(\mathbf{a}, \mathbf{a'},\mathbf{e}, \mathbf{e'})$ and $(M_1, M_2, M_3)$ as follows.
    \begin{itemize}
        \item If $u,v\in R$, then push vertex $u,v$ into $\mathbf{a}$, and let $(u, v)$ be a matching pair in $M_1$.
        \item If $u\in R$ and $v \in T$, then push vertex $u$ into $\mathbf{a}$ and $v$ into $\mathbf{a}'$, and let $(u,v)$ be a matching pair in $M_2$.
        \item If $u,v\in T$, then push $u,v$ into $\mathbf{e}'$, and let $(u,v)$ be a matching pair in $M_3$.
    \end{itemize}
    Clearly triple $(M_1,M_2,M_3)$ is a perfect matching of triple $(R_{\mathbf{a}}, T_{\mathbf{a}'}, G_{\mathbf{e}, \mathbf{e'}})$.
\end{itemize}
Thus we are done!
\end{proof}

\section{Proof of Theorem~\ref{thm:marginal}} \label{app:proof_theorem_1}

\begin{proof}[Proof of Theorem~\ref{thm:marginal}]

For our purpose, we use $\{ R_1, \ldots, R_k \}$ to relabel those variables in $R$, and $\{ T_1, \ldots, T_{n-k} \}$ to relabel those in $T$. Similar notation applies to variables in $R^*$ and $T^*$. In one term of the original expansion of $q_n(x_1, \ldots, x_k)$ (Eq.~\eqref{eq:marginal}), suppose there are $j_1$ edges of the form $(R_i, R_j)$ and $j_2$ edges of the form $(R^*_i, R^*_j)$. Recall that Rule ~1 tells us an edge is equivalent to a matrix element. The values of $j_1$ and $j_2$ may vary in different terms. 
Then the number of other kinds of matching pairs are as shown in Table~\ref{tab:MatchThm}.
\begin{table}[h]
\centering
\begin{tabular}{c|c|c|c}
\hline\hline
 & case 1 & case 2 & case 3 \\ \hline
$W$   & \begin{tabular}[c]{@{}c@{}}$(R_i, R_l)$\\ $j_1$ pairs\end{tabular} & \begin{tabular}[c]{@{}c@{}}$(R_i, T_l)$\\ $k - 2j_1$ pairs\end{tabular} & \begin{tabular}[c]{@{}c@{}}$(T_i, T_l)$\\ $\frac{n}{2} - k + j_1$ pairs\end{tabular} \\ \hline
$W^*$ & \begin{tabular}[c]{@{}c@{}}$(R^*_i, R^*_l)$\\ $j_2$ pairs\end{tabular}     & \begin{tabular}[c]{@{}c@{}}$(R^*_i, T^*_l)$\\ $k - 2j_2$ pairs\end{tabular}   & \begin{tabular}[c]{@{}c@{}}$(T^*_i, T^*_l)$\\ $\frac{n}{2} - k + j_2$ pairs\end{tabular} \\ \hline\hline
\end{tabular}
\caption{Number of pairs of different cases in one term of Eq.~\eqref{eq:marginal}.}
\label{tab:MatchThm}
\end{table}

Note that $T_i$ and $T^*_i$ are from the set $X_n\backslash X_k$, which will be summed over, and $R_i$ and $R^*_i$ are from the set $X_k$. From Lemma~\ref{lem:ThmSplitHaf}, one can see that $\Haf(W_{\x})$ will give $\Haf(R_{\a})$ and $\Haf(W^*_{\x})$ will give $\Haf(R^*_{\a'})$. Both terms contain no vertices from $T$ and $T^*$, so they are retained in the final expression of $q_n(x_1, \cdots, x_k)$.

\begin{figure*}
    \centering
  \begin{tikzpicture}[shorten >=1pt, auto, node distance=3cm, ultra thick,scale=0.8]
    \tikzstyle{node_style} = [circle,draw=gray,shade,minimum size=2pt,
    font=\sffamily\bfseries]
    \tikzstyle{node_style2} = [
  text=gray, font=\sffamily\bfseries]
    \tikzstyle{node_style3} = [circle]
    \tikzstyle{edge_style} = [draw=black, line width=2, ultra thick] 

\draw (0,0) ellipse (3cm and 1.5cm);
\draw (-1.2,-0.5) ellipse(1cm and 0.5cm);
\draw (1,0.5) ellipse(1cm and 0.5cm);

\node[node_style2]  at (-1.8,-0.5) {$A$};
\node[node_style2]  at (-1.6,0.8) {$\textbf{a}$};
\node[node_style2]  at (0.5,0.5) {$\textbf{e}$};
\node[node_style2]  at (0,-2) {$R$};

\node (a) at (1.5,0.6) {};
\node (a2) at (1.1, 0.6) {};
\node (b) at (-1.4,-0.4) {};

\fill (a) circle (3pt);
\fill (a2) circle (3pt);
\fill (b) circle (3pt);

\draw (9,0) ellipse (3cm and 1.5cm);
\draw (10.5,0) ellipse(1cm and 0.5cm);
\draw (8,-0.6) ellipse(1cm and 0.5cm);
\draw (8, 0.6) ellipse(1cm and 0.5cm);

\node (c) at (8,0.6) {};
\node (c2) at (8.5, 0.6) {};
\node (d) at (8.5,-0.4) {};
\node (d2) at (8.5, -0.8){};
\node (e) at (11,0) {};

\fill (c) circle (3pt);
\fill (c2) circle (3pt);
\fill (d) circle (3pt);
\fill (d2) circle (3pt);
\fill (e) circle (3pt);

\node[node_style2]  at (10,0) {$T_a$};
\node[node_style2] at (7.5,0.6) {$T_b$};
\node[node_style2] at (7.7,-0.7) {$T_c$};
\node[node_style2] at (9,-2) {$T$};
\node[node_style2]  at (9.5,-1) {$T_d$};

\draw (0,-6) ellipse (3cm and 1.5cm);
\draw (-1.5,-6.5) ellipse(1cm and 0.5cm);
\draw (1,-5.6) ellipse(1cm and 0.5cm);

\node[node_style2]  at (-1,-5) {$\textbf{a}'$};
\node[node_style2]  at (-2,-6.4) {$B$};
\node[node_style2]  at (0.5,-5.5) {$\textbf{e}'$};
\node[node_style2]  at (0,-8) {$R^*$};
\node(a1) at (1.5,-5.6) {};
\node(a3) at (1,-5.6) {};
\node (b1) at (-1.5,-6.6) {};

\fill (a1) circle (3pt);
\fill (a3) circle (3pt);
\fill (b1) circle (3pt);

\draw (9,-6) ellipse (3cm and 1.5cm);
\draw (10.5,-6) ellipse(1cm and 0.5cm);
\draw (8,-6.6) ellipse(1cm and 0.5cm);
\draw (8, -5.5) ellipse(1cm and 0.5cm);

\node(c1) at (8,-5.6) {};
\node(c3) at (8.5,-5.6) {};
\node (d1) at (8.5,-6.4) {};
\node (d3) at (8.5, -6.8) {};
\node (e1) at (11,-6) {};

\fill (c1) circle (3pt);
\fill (d1) circle (3pt);
\fill (e1) circle (3pt);
\fill (d3) circle (3pt);
\fill (c3) circle (3pt);

\node[node_style2]  at (10,-6) {$T_a^*$};
\node[node_style2] at (7.5,-6.7) {$T_c^*$};
\node[node_style2] at (7.5,-5.5) {$T_b^*$};
\node[node_style2] at (9,-8) {$T^*$};
\node[node_style2]  at (9.5,-7) {$T_d^*$};

\begin{scope}[
              every node/.style={fill=white,circle},
              every edge/.style={draw=airforceblue,very thick}]
\path [-] (a) [bend left=0] edge node[right] {(b)} (c);
\path [-] (a2) [bend left=10] edge node[right]  {} (c2);
\path [-] (b) [bend left=30]edge node[right]  {(a)} (e);
\path [-] (a1) [bend left=0] edge node[right] {} (d1);
\path [-] (a3) [bend left = 0]edge node[right]  {(c)} (d3);
\path [-] (b1) [bend left=-25] edge node[right]  {(a)} (e1);
\end{scope}
\draw [airforceblue] (8,-5.6) --(8.5,-5.6);
\draw [airforceblue] (8.5, -0.4) -- (8.5, -0.8);
\path [<->,dashed,gray] (c) [bend left=-30] edge node[above]{}  (c1);
\path [<->,dashed,gray] (c2) [bend left=-20] edge node[above]{}  (c3);
\path [<->,dashed,gray] (d2) [bend left=30] edge node[above]{}  (d1);
\path [<->,dashed,gray] (d) [bend left=40] edge node[above]{}  (d3);

\path [<->,dashed,gray] (e) [bend left=0] edge node[above]{}  (e1);
\end{tikzpicture}
    \caption{Visualization of Summation on $T$ and $T^*$.}
    \label{fig:VisualThm}
\end{figure*}
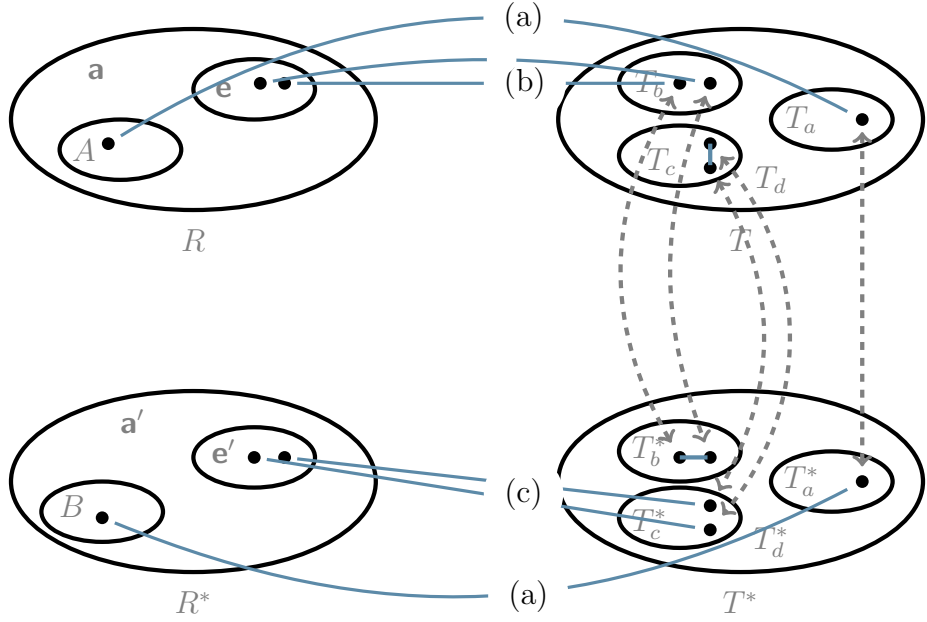

Now we want to analyze what the four cases in Lemma~\ref{lem:summation_path} result in. Suppose in one term, there are $\mu$ summation paths of case~3 (again $\mu$ may vary in different summands). Then $\mu \leq \min(k - 2j_1, k - 2j_2)$ since there is at least one edge $(R_i, T_l)$ and $(R^*_i, T_l)$ in the expression. 
By counting the number of pair $(R_i, T_l)$ (or $(R^*_i, T_l)$), one obtains that the number of paths of case~1 and case~2 is $(k - 2j_1 - \mu)/2$ and $(k - 2j_2 - \mu)/2$, respectively. Since $(k - 2j_1 - \mu)/2$ is an integer, $\mu$ must be of the same parity as $k$, that is $\mu \equiv k \text{ mod 2}$.
Furthermore, there is at least one pair of $(T_i, T_l)$ in case~1 and~2, so
\ba
\frac{k - 2j_2 - \mu}{2} \leq \frac{n}{2} - k + j_2 \Longrightarrow \mu \geq 3k - n - 2(j_1 + j_2) \ .
\ea
So the range of $\mu$ is given by,
\ba
S_{\mu} \equiv \{\mu \in \mathbb{N}: (3k-2(j_{1}+j_{2})-n) \leq \mu \leq k - 2\max(j_{1},j_{2}), k \equiv \mu \text{ mod 2} \} \ .
\ea

For case~3, when $\sum_{\sigma, \tau \in \cal{M}_n}$ is taken into account, what we are actually doing is first multiply some $\delta(*, *)$ together and then take the summation over something. Compactly, if we let $A \in \binom{X_k\backslash \a}{\mu}$ and $B \in \binom{X_k\backslash \a'}{\mu}$, and define $S_{A,B}$ as $S_{A,B}(i, j) \equiv \delta({A_i}, {B_j})$, then case~3 gives
\ba
\sum_{\mu \in S_{\mu}} \sum_{\substack{ A \in \binom{X_k\backslash \a}{\mu} \\ B \in \binom{X_k\backslash \a'}{\mu} }} \Per(S_{A,B}) \ .
\ea
multiplied by a parameter which is related to the selected matchings of $T_i$ in case (a).
Case~1-3 together give
\ba
\sum_{\mu \in S_{\mu}} \sum_{\substack{ A \in \binom{X_k\backslash \a}{\mu} \\ B \in \binom{X_k\backslash \a'}{\mu} }} \Per(S_{A,B}) \sum_{\substack{\e = X_k\backslash\{ \a\cup A \} \\ \e' = X_k\backslash\{ \a'\cup B \} }} \Haf(R_{\e}) \Haf(R^*_{\e'})  \ .
\ea
Putting all things together, the final expression of $q(x_1, \cdots, x_k)$ is 
\begin{align*}
q_n(x_1, \cdots, x_k) =&\frac{1}{f_n}|\Haf{(W_\mathbf{x})}|^2\\
=&\frac{(n-k)!}{f_n}\sum_{j_1, j_2 = \max(0, k - \frac{n}{2})}^{\floor{\frac{k}{2}}} \sum_{\substack{\a \in \binom{X_k}{2j_1} \\ \a' \in \binom{X_k}{2j_2}}} \Haf(R_{\a}) \Haf(R^*_{\a'}) \sum_{\substack{\mu \in S_{\mu}  }} F(k, \mu, j_1, j_2) \notag \\
&  \sum_{\substack{ A \in \binom{X_k\backslash \a}{\mu} \\ B \in \binom{X_k\backslash \a'}{\mu} }} \Per(S_{A,B}) \sum_{\substack{ \e = X_k\backslash\{ \a\cup A \} \\ \e' = X_k\backslash\{ \a'\cup B \} }} \Haf(R_{\e}) \Haf(R^*_{\e'})  \ ,
\end{align*}
where $F(k, \mu, j_1, j_2)$ is related to the matching patterns of $T_i$, which is given in Appendix \ref{ExpressionF}.
\end{proof}

\section{\label{ExpressionF}Expression of the factor $F(k, \mu, j_1, j_2)$}

This section explains why $F(k, \mu, j_1, j_2)$ equals $(n-k)!\binom{\frac{n-k+\mu+m}{2}-1}{k-j_{1}-j_{2}+\frac{m}{2}-1}$.
$F(k, \mu, j_1, j_2)$ is 
the number of all of matchings corresponding to $x_i$ 
in which $i>k$, in summation paths of case 1-3 and with form
\begin{align}
     \sum_{x_{i_{1}}\cdots x_{i_{s}}\in [m]} W(x_{i_1}, x_{i_{2}}) W^*(x_{i_2}, x_{i_3})\cdots W^*(x_{i_s}, x_{i_1})
    \label{eq:type_d}
\end{align}
When we consider all of internal perfect matchings, which gives 
\begin{equation}
    \sum_{x_{i_{1}}\cdots x_{i_{s}}\in [m]}|\text{Haf}(W_{x_{i_{1}},\cdots,x_{i_{s}}})|^{2}=\binom{\frac{m+s}{2}-1}{\frac{s}{2}} s! = f_{s}.
    \label{eq:type_dmatch}
\end{equation}
We label $W(x_{i_{a}}, x_{i_{b}})$ as $({i_{a}}, {i_{b}})$, and label $W^{*}(x_{i_{a}}, x_{i_{b}})$ as $({i_{a}}, {i_{b}})^{*}$ for convenience. Observe that there are equal amount of $({i_{a}},{i_{b}})$ and $({i_{a'}},{i_{b'}})^{*}$ in a summation path of case~3, and there are one more $({i_{a}}, {i_{b}})^{*}$ in case~1, one more $({i_{a}}, {i_{b}})$ in case~2. Suppose there are $c_{1}$ matching pairs $({i_{a}},{i_{b}})$ for $\mu$ summation paths of case~3, $c_{2}$ matching pairs $({i_{a}},{i_{b}})$ for $(k-2j_{1}-\mu)/2$ summation paths of case~1, and $c_{3}$ matching pairs $({i_{a}},{i_{b}})^{*}$ for $(k-2j_{2}-\mu)/2$ summation paths of case~2, and the remaining $d$ matching pairs $({i_{a}},{i_{b}})$ which generate case~4. Let $C_{max}:= (c_{1} + c_{2} + c_{3} + d)/2$, then $C_{max} = (n-3k+2j_{1}+2j_{2}+\mu)$ since the summation on last column of table \ref{tab:sequence} equals to $n - k$.

\begin{table}[h]
    \centering
    \begin{tabular}{c|c|c |c}
    \hline \hline
    Summation path     &  $\#$ chains & $(i_a, i_b)$ pairs & $|\{i_a|1\leq a \leq n - k\}|$ \\
    \hline
    case~3 & $\mu$ & $c_1$ & $2c_1 + \mu$\\
\hline
case~1 & $\frac{k - 2j_1 - \mu}{2}$ & $c_2$ & $2c_2 + k - 2j_1 - \mu$\\
  \hline
  case~2 & $\frac{k - 2j_2 - \mu}{2}$ & $c_3 - 1$ &  $2c_3 + k - 2j_2 - \mu$\\
\hline
  case~4  &  &  $d$ & $2d$\\
    \hline\hline
    \end{tabular}
    \caption{Relationship between $i_a$ and the number of pairs and chains for each sequence type.}
    \label{tab:sequence}
\end{table}

Suppose there are $u_{1},\cdots,u_{\mu}$ matching pairs for the $\mu$ summation paths of case~3 respectively. Thus $u_{1}+\cdots + u_{\mu} = c_{1}$, and the number of the matchings equals to\footnote{$n^{\underline{k}} = \binom{n}{k}k!$.}
\begin{align}
    \binom{c_{1}+\mu-1}{c_{1}}(n-k)^{\underline{2c_{1}+\mu}}.
    \label{eq:matchingA}
\end{align}
The first item $\binom{c_1 + \mu - 1}{c_1}$ of Eq. \eqref{eq:matchingA} gives all of species of $u_1, \cdots, u_{\mu}$. In other hand, for case~3 which has $l$ $(i_a, i_b)$ pairs, the overall matchings equals to $(l + 1)!$, thus for all of $2c_1$ pairs, $\mu$ chains and fixed $u_1, \cdots, u_{\mu}$, there will be $(n-k)^{\underline{2c_{1}+\mu}}$ matchings.

In the same way, all of the matchings for $k - 2j_{1}-\mu$ summation paths of case~1 equals to
$$\binom{c_{2} + \frac{k - 2j_{1} - \mu}{2}-1}{c_{2}}(n - k - 2c_{1}-\mu)^{\underline{2c_{2}+k-2j_{1}-\mu}},$$
all of the matchings for $k - 2j_{2} - \mu$ summation paths of case~2 equals to
$$\binom{c_{3} + \frac{k - 2j_{2} - \mu}{2}-1}{c_{3}}(n-2k-2c_{1} - 2c_{2} +2j_{1})^{\underline{2c_{3}+k-2j_{2}-\mu}},$$
all of the matchings for the remaining $d$ matching pairs of case~4 equals to
$f_{2d}$. Thus all of the matchings for $x_{i}$, where $x_{i}>k$, equals to
\begin{align*}
&(n-k)!\sum_{c_{1}=0}^{\frac{C_{max}}{2}}\sum_{c_{2}=0}^{\frac{C_{max}}{2}-c_{1}}\sum_{c_{3}=0}^{\frac{C_{max}}{2}-c_{1}-c_{2}}\binom{c_{1}+\mu-1}{\mu-1}\binom{c_{2} + \frac{k - 2j_{1} - \mu}{2} - 1}{c_{2}}\binom{c_{3} + \frac{k - 2j_{2} - \mu}{2} - 1}{c_{3}}\binom{\frac{m}{2} + d - 1}{d}\\
&=(n-k)!\binom{\frac{n-k+\mu+m}{2}-1}{k-j_{1}-j_{2}+\frac{m}{2}-1}.
\end{align*}
The equation holds by Eq. \eqref{eq:Concrete} \cite{graham1989concrete}
\begin{equation}
    \sum_{0\leq k \leq l} 
    \binom{l - k}{m} \binom{q + k}{n} = \binom{l + q + 1}{ m + n + 1} \text{, where } l,m\geq 0, n\geq q \geq 0.
    \label{eq:Concrete}
\end{equation}



\section{Brute-force sampling.\label{app:brute-force}}

In the brute-force sampling algorithm, we work on $p_n(\s)$ instead of $q_n(\x)$. The algorithm works as follows. First, we divide the interval $[0, 1]$ into $\binom{m+n-1}{n}$ number of intervals (which is the number of $\s$), where the length of each interval is given by $p_n(\s)$. So each interval is related to an $\s$. Then we sample a value $w$ randomly from $[0, 1]$, and the algorithm outputs the corresponding $\s$ of the interval that $w$ lies in. The probability that the algorithm outputs $\s$ is the length of the interval, that is $p_n(\s)$. 
The worst time complexity is $\binom{m + n - 1}{n} T(\text{Haf}^n)$, and in most case we need to compute $c\binom{m + n - 1}{n}$ Hafnians with input size $n$, where $c$ is constant in $(0,1]$. In Figure~\ref{fig:comparison} (b) we set $c = 0.1$ for estimating the time of brute force sampling. For example, when there are $n = 8$ photons and $m = n^2$ modes, we need to compute at least $\frac{1}{10}\binom{71}{8}\approx 1.06\times 10^9$ Hafnians whose size is $8$ with high probability. The time of the algorithm is blowing up for a little input photons $n$ by the number of Hafnians to be computed, so we need to search for a new algorithm.

\section{More numerical results}
\label{sec:more_numerical}

Figure~\ref{fig:comparison_cd} shows a comparison of actual implementation with the theoretical bounds in Appendix~\ref{timeA} and~\ref{timeB}.

\begin{figure*}
\center
\includegraphics[trim = 0mm 40mm 0mm 50mm, clip=true,width = 1.0\textwidth]{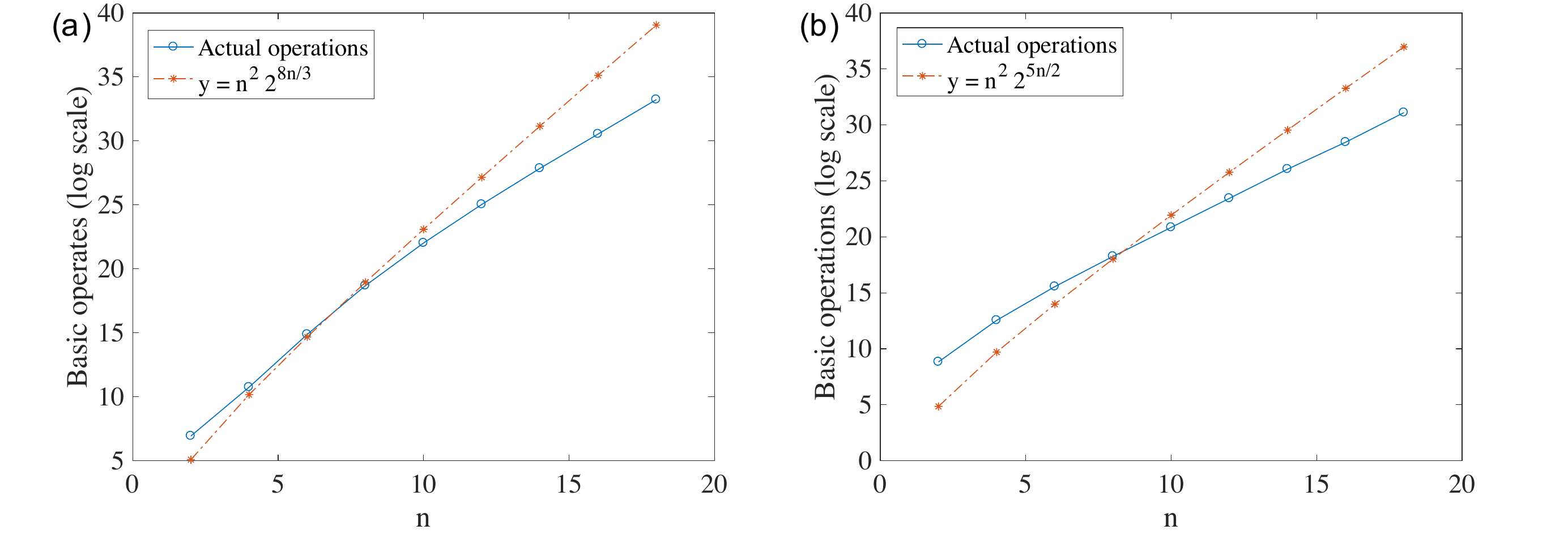}
\caption{(a) Comparison of basic operations counting from actual implementation for the polynomial-space algorithm with the upper bound in time complexity analysis. (b) Comparison of basic operations counting from actual implementation for the exponential-space algorithm with the upper bound in time complexity analysis.
}
\label{fig:comparison_cd}
\end{figure*}


\end{appendix}

\end{document}